\documentclass[tbtags]{lmcs}
\pdfoutput=1

\usepackage{lastpage}
\renewcommand{\dOi}{14(3:9)2018}
\lmcsheading{}{1--\pageref{LastPage}}{}{}%
{Dec.~12,~2017}{Aug.~20,~2018}{}

\usepackage{amssymb}
\usepackage{amsfonts}
\usepackage{mathrsfs}
\usepackage{latexsym}
\usepackage{ifthen}
\usepackage{url}
\usepackage{verbatim}
\usepackage{xspace}
\usepackage{stmaryrd}
\usepackage{braket}
\usepackage{enumitem}

\newcommand{\DecTh}{\mathsf{DecTh}}
\newcommand{\FOMOD}{\mathsf{FO}\!+\!\mathsf{MOD}}
\newcommand{\FO}{\mathsf{FO}}
\newcommand{\MSOTh}{\mathsf{MTh}}
\newcommand{\MSO}{\mathsf{MSO}}
\newcommand{\N}{\mathbb N}
\newcommand{\Rec}{\mathsf{REC}}
\newcommand{\Sent}{\mathsf{Sent}}
\newcommand{\Tot}{\mathsf{TOT}}
\newcommand{\UndecTh}{\mathsf{UndecTh}}
\newcommand{\Z}{\mathbb Z}
\newcommand{\dom}{\mathsf{dom}}
\newcommand{\indrec}{\mathrm{rec}}
\newcommand{\qr}{\mathsf{qr}}
\newcommand{\typefun}{\mathrm{tp}}
\newcommand{\windrec}{\mathrm{rec'}}

\begin{document}
\title{Infinite and Bi-infinite Words with Decidable Monadic Theories}
\author{Dietrich Kuske}
\address{Institut f\"ur Theoretische Informatik, TU Ilmenau, Germany}
\email{dietrich.kuske@tu-ilmenau.de}
\author{Jiamou Liu}
\address{The University of Auckland, New Zealand}
\email{jiamou.liu@auckland.ac.nz}
\author{Anastasia Moskvina}
\address{Auckland University of Technology, New Zealand}
\email{anastasia.moskvina@aut.ac.nz}
\thanks{This work was partially supported by the Marsden fund of New
  Zealand AUT1101.}

\begin{abstract}
  We study word structures of the form $(D,<,P)$ where $D$ is either
  $\N$ or $\Z$, $<$ is the natural linear ordering on $D$ and
  $P\subseteq D$ is a predicate on $D$. In particular we show:
  \begin{enumerate}[label=(\alph*)]
  \item The set of recursive $\omega$-words with decidable monadic
    second order theories is $\Sigma_3$-complete.
  \item Known characterisations of the $\omega$-words with decidable
    monadic second order theories are transfered to the corresponding
    question for bi-infinite words.
  \item We show that such ``tame'' predicates $P$ exist in every
    Turing degree.
  \item We determine, for $P\subseteq\Z$, the number of predicates
    $Q\subseteq\Z$ such that $(\Z,\le,P)$ and $(\Z,\le,Q)$ are
    indistinguishable by monadic second order formulas.
  \end{enumerate}

  Through these results we demonstrate similarities and differences
  between logical properties of infinite and bi-infinite words.
 \end{abstract}
\maketitle

\section{Introduction}

The decision problem for logical theories of linear structures and
their expansions has been an important question in theoretical
computer science.  B\"uchi in \cite{Buchi} proved that the monadic
second order theory (henceforth ``MSO-theory'') of the linear ordering
$(\N, \le)$ is decidable. Expanding the structure $(\N, \le)$ by unary
functions or binary relations typically leads to undecidable
MSO-theories. Hence numerous works have been focusing on structures of
the form $(\N, \le, P)$ where $P$ is a unary predicate.  Elgot and
Rabin \cite{ElgotRabin} showed that for many natural unary predicates
$P$, such as the set of factorial numbers, the set of powers of $k$,
and the set of $k$th powers (for fixed $k$), the structure
$(\N, \le, P)$ has decidable MSO-theory; on the other
hand, there are structures $(\N, \le, P)$ whose MSO-theory is
undecidable \cite{BuchiLand}. Many subsequent works further expanded
the field
\cite{Siefkes,CartonThomas,Semenov,Semenov2,RabThomas,Rabinovich}.

\begin{enumerate}
\item Semenov observed that, in order to decide the MSO-theory of an
  $\omega$-word $\alpha=(\N,\le,P)$, it suffices to determine whether
  a given regular language has a factor in $\alpha$ beyond every
  position and, if not, determine some position beyond which $\alpha$
  has no factor from that language. This idea is formalised by the
  notion of an ``indicator of recurrence''.  In \cite{Semenov}, he
  provided a full characterisation: $(\N, \le , P)$ has decidable
  MSO-theory if and only if $P$ is recursive and there is a recursive
  indicator of recurrence for $P$.

\item Rabinovich and Thomas observed that, in order to decide the
  MSO-theory of an $\omega$-word $\alpha=(\N,\le,P)$, it suffices to
  have some computable factorisation of $\alpha$ such that late
  factors cannot be distinguished by formulas of bounded quantifier
  depth. This idea is formalised by the notion of a ``uniformly
  homogeneous set''. In \cite{RabThomas}, Rabinovich and Thomas
  provided a full characterisation: $(\N, \le, P)$ has a decidable
  MSO-theory if and only if $P$ is recursive and there is a recursive
  uniformly homogeneous set.

  They also observed that it suffices to be able to compute, from a
  quantifier depth $k$, a pair of $k$-types $(u,v)$ such that $\alpha$
  can be factorised in such a way that all factors (except the first
  factor that has type $u$) have type $v$.  In \cite{RabThomas}, they
  also showed that the MSO-theory of $(\N,\le,P)$ is decidable if and
  only if there is a ``recursive type-function''.
\end{enumerate}
This paper has three general goals: The first is to compare these
characterisations in some precise sense. The second is to investigate
the above results in the context of \emph{bi-infinite words}, which
are structures of the form $(\Z, \le, P)$. The third is to compare the
logical properties of infinite words and bi-infinite words. More
specifically, the paper discusses the following questions:

\begin{enumerate}[label=(\alph*)]
\item In Section~\ref{sec:complexity}, we analyze the
  recursion-theoretical bound of the set of all computable predicates
  $P\subseteq \N$ where $(\N, \le, P)$ has a decidable MSO-theory.  It
  is noted that the second characterisation by Rabinovich and Thomas
  turns out to be a $\Sigma_5$-statement. Differently the
  characterisation by Semenov and the first characterisation by
  Rabinovich and Thomas both consist of $\Sigma_3$-statements, and
  hence deciding if a given $(\N, \le, P)$ has decidable MSO-theory is
  in $\Sigma_3$. We show that the problem is in fact
  $\Sigma_3$-complete. Hence these two characterisations are optimal
  in terms of their recursion-theoretical complexity.

\item If the MSO-theory of $(\N,\le,P)$ is decidable,
  then $P$ is recursive. For bi-infinite words of the form
  $(\Z,\le,P)$, this turns out not to be necessary. To the contrary,
  Theorem~\ref{thm:non-recursive-but-decidable-MSOTh} demonstrates
  that every m-degree contains some bi-infinite word with a decidable
  theory. Even more, Corollary~\ref{cor:Turing-degrees} shows that every
  decidable MSO-theory of a recurrent bi-infinite word is realised
  in every Turing degree.

\item In the rest of Section~\ref{sec:bi-infinite}, we then
  investigate which of the three characterisations can be lifted to
  bi-infinite words, i.e., structures of the form $(\Z,\le,P)$ with
  $P\subseteq\Z$. It turns out that this is nicely possible for
  Semenov's characterisation and for the second characterisation by
  Rabinovich and Thomas, but not for their first one.

\item The final Section~\ref{sec:counting} investigates how many
  bi-infinite words are indistinguishable from $(\Z,\le,P)$. It turns
  out that this depends on the periodicity properties of $P$: if $P$
  is periodic, there are only finitely many equivalent bi-infinite
  words, if $P$ is recurrent and non-periodic, there are
  $2^{\aleph_0}$ many, and if $P$ is not recurrent, then there are
  $\aleph_0$ many.
\end{enumerate}

\section{Preliminaries} \label{sec:prelim}

\subsection{Words}

We use $\N$, $\widetilde{\N}$ and $\Z$ to denote the set of natural
numbers (including $0$), negative integers (not containing $0$), and
integers, respectively.

A \emph{finite word} is a mapping $u\colon\{0,1,\dots,n-1\}\to\{0,1\}$
with $n\in\N$, it is usually written $u(0) u(1) u(2) \cdots
u(n-1)$. The set of positions of $u$ is $\{0,1,\dots,n-1\}$, its
length $|u|$ is $n$. The unique finite word of length $0$ (and
therefore with empty set of positions) is denoted $\varepsilon$. The
set of all finite words is $\{0,1\}^*$, the set of all non-empty
finite words is $\{0,1\}^+$.

An \emph{$\omega$-word} is a mapping from $\N$ into $\{0,1\}$. Often,
an $\omega$-word $\alpha$ is written as the sequence
$\alpha(0)\alpha(1)\alpha(2)\cdots$. Its set of positions is $\N$;
$\{0,1\}^\omega$ is the set of $\omega$-words. An
\emph{$\omega^*$-word} is a mapping $\alpha$ from $\widetilde{\N}$
into $\{0,1\}$, it is usually written as the sequence
$\cdots \alpha(-3) \alpha(-2) \alpha(-1)$. Its set of positions is
$\widetilde{\N}$; $\{0,1\}^{\omega^*}$ is the set of $\omega^*$-words.

Finally, a \emph{bi-infinite word} $\xi$ is a mapping from $\Z$
into $\{0,1\}$, written as the sequence $$\cdots \xi(-2) \xi(-1)
\xi(0) \xi(1) \xi(2) \cdots$$ (this notation has to be taken
with care since, \textit{e.g.}, the bi-infinite words
$\xi_i\colon\Z\to\{0,1\}\colon n\mapsto (|n|+i)\bmod 2$ with
$i\in\{0,1\}$ are both described as $\cdots 0 1 0 1 0 1 0\cdots$, but
they are different). The set of positions of a bi-infinite word is
$\Z$.

Shift-equivalence and period will be important notions in this
context: two bi-infinite words $\xi$ and $\zeta$ are
\emph{shift-equivalent} if there is $p\in\N$ with $\xi(n)=\zeta(n+p)$
for all $n\in\Z$. Furthermore, the period of the bi-infinite word
$\xi$ is the least natural number $p>0$ with $\xi(n)=\xi(n+p)$ for all
$n\in\N$ -- clearly, the period needs not exist:
$0^{\omega^*}\,1^\omega$ is not periodic.

When saying ``word'', we mean ``a finite, an $\omega$-, an $\omega^*$-
or a bi-infinite word'', ``infinite word'' means ``$\omega$- or
$\omega^*$-word''.

For two finite words $u$ and $v$, the concatenation $uv$ is again a
finite word of length $|u|+|v|$ with $uv(i)=u(i)$ if $0\le i<|u|$ and
$uv(i)=v(i-|u|)$ for $|u|\le i<|uv|$. More generally, and in a similar
way, we can also concatenate a finite or $\omega^*$-word $u$ and a
finite or $\omega$-word $v$ giving rise to some word $uv$. Similarly,
we can concatenate infinitely many finite nonempty words $u_i$ giving
an $\omega$-word $u_0u_1u_2\cdots$, an $\omega^*$-word $\cdots u_{-2}
u_{-1} u_0$, and a bi-infinite word $\cdots u_{-2} u_{-1}
u_0u_1u_2\cdots$ (where the position $0$ is the first position of
$u_0$). As usual, $u^\omega$ denotes the $\omega$-word $uuuu\cdots$
for $u\in\{0,1\}^+$, analogously, $u^{\omega^*}=\cdots u u u $.

Let $w$ be some word and $i,j$ be two positions with $i\le j$. Then we
write $w[i,j]$ for the finite word
$w(i)w(i+1)\cdots w(j)\in\{0,1\}^+$. A finite word $u$ is a
\emph{factor} of $w$ if there are $i,j\in D$ with $i\le j$ and
$w[i,j]=u$ or if $u$ is the empty word $\varepsilon$. The set of
factors of $w$ is $F(w)$. If $\beta$ is an $\omega$- or a
bi-infinite word and $i$ is a position in $w$, then $\beta[i,\infty)$
is the $\omega$-word $\beta(i) \beta(i+1)
\beta(i+2)\cdots$. Symmetrically, if $\beta$ is an $\omega^*$- or a
bi-infinite word and $i$ is a position in $\beta$, then
$\beta(-\infty,i]$ is the $\omega^*$-word $\cdots \beta(i-2) \beta(i-1)
\beta(i)$.

Let $u$ be some finite word. Then $u^R$ is the \emph{reversal} of $u$,
i.e., the finite word of length $|u|$ with $u^R(i)=u(|u|-i-1)$ for all
$0\le i<|u|$. The reversal of an $\omega$-word $\alpha$ is the
$\omega^*$-word $\alpha^R$ with $\alpha^R(i)=\alpha(-i-1)$ for all
$i\in\widetilde\N$. The reversal of an $\omega^*$-word $\alpha$ is the
$\omega$-word $\alpha^R$ with $\alpha^R(i)=\alpha(-i-1)$ for all
$i\in\N$. Finally, the reversal of a bi-infinite word $\xi$ is the
bi-infinite word $\xi^R$ with $\xi^R(i)=\xi(-i)$ for all $i\in\Z$.

An $\omega$-word $\beta$ is \emph{recurrent} if any factor of $\beta$
appears infinitely often in $\beta$, i.e.,
$F(\beta)=F(\beta[i,\infty))$ for all $i\in\N$. A bi-infinite word
$\xi$ is \emph{recurrent} if
$F(\xi)=F(\xi(-\infty,i])=F(\xi[i,\infty))$ for all $i\in\Z$. Note
that $0^{\omega^*}1^\omega$ is not recurrent despite the fact that it
is build from two recurrent $\omega$-words. Semenov \cite{Semenov}
calls recurrent bi-infinite words \emph{homogeneous}.

\subsection{Logic}

With any word $w$, we associate a relational structure $M_w=(D,\le,P)$
where $D$ is the set of positions of $w$ (i.e., a subset of $\Z$),
$\le$ is the restriction of the natural linear order on $\Z$ to $D$,
and $P=\{n\in D\mid w(n)=1\}=w^{-1}(1)$. Relational structures of this
form are called labeled linear orders.

We use the standard logical system over the signature of labeled
linear orders. Hence first order logic $\FO$ has relational symbols
$\le$ and $P$.  Monadic second order logic $\MSO$ extends $\FO$ by
allowing unary second order variables $X,Y,\ldots$, their
corresponding atomic predicates (\textit{e.g.} $X(y)$), and quantification over
set variables. By $\Sent$, we denote the set of sentences of the logic
$\MSO$.

For a word $w$ and an $\MSO$-sentence $\varphi$, we write
$w\models\varphi$ for ``the sentence $\varphi$ holds in the relational
structure $M_w$''. The \emph{$\MSO$-theory} of the word $w$ is the set
$\MSOTh(w)$ of all $\MSO$-sentences $\varphi$ that are true in~$w$.

\begin{exa}\label{exa-2.1}
  We use $\mathrm{succ}(x,z)$ as shorthand for the formula
  $x<z\land\lnot\exists y\colon x<y<z$. Let $n\in\N$ and $\varphi(x)$
  be a formula with a free variable $x$. Then the formula
  $$\exists x_0,x_1,\dots,x_n\colon \bigwedge_{0\le
    i<n}\mathrm{succ}(x_i,x_{i+1})\land x_n= x\land \varphi(x_0)$$
  expresses that $\varphi$ holds for $x-n$. We will abbreviate this as
  $\varphi(x-n)$.

  Let $n\in\N$ and consider the following formula:
  \[
    \varphi(x,y)= \exists X\colon
      \forall z\colon (X(z)\Leftrightarrow z=x\lor (x<z\land X(z-n)))\land X(y)
  \]
  If $w$ is some word with two positions $i$ and $j$, then
  $w\models\varphi(i,j)$ if and only if $i\le j$ and $n\mid j-i$.
\end{exa}

With any $\MSO$-formula $\varphi$, we associate its \emph{quantifier
  rank} $\qr(\varphi)\in \N$: the atomic formulas have quantifier rank
0; $\qr(\varphi_1\land\varphi_2)=\qr(\varphi_1\lor\varphi_2)=
\max\{\qr(\varphi_1),\qr(\varphi_2)\}$;
$\qr(\lnot\varphi)=\qr(\varphi)$; and $\qr(\exists
X\colon\varphi)=\qr(\forall X\colon\varphi)=\qr(\varphi)+1$ where
$X$ is a first- or second-order variable.
\begin{defi}
  Let $k\in \N$. Two words $w_1$ and $w_2$ are \emph{$k$-equivalent}
  (denoted $w_1\equiv_k w_2$) if $w_1\models\varphi$ if and only if
  $w_2\models\varphi$ for all $\MSO$-sentences $\varphi$ with
  $\qr(\varphi)\leq k$.  Equivalence classes of this equivalence
  relation are called \emph{$k$-types}.

  The words $w_1$ and $w_2$ are \emph{$\MSO$-equivalent} (denoted
  $w_1\equiv w_2$) if $w_1\equiv_k w_2$ for all $k\in\N$.  Equivalence
  classes of this equivalence relation are called \emph{types}.
\end{defi}

Let $k\ge2$ and $u,v$ be two words with $u\equiv_k v$. If $u$ is
finite, then it satisfies the sentence $(\exists x\forall y\colon x\le
y)\land(\exists x\forall y\colon x\ge y)$. Consequently, also $v$ is
finite. Analogously, $u$ is an $\omega$-word iff $v$ is an
$\omega$-word etc. We will therefore speak of a ``$k$-type of finite
words'' when we mean a $k$-type that contains some finite word (and
analogously for $\omega$-words etc).

Often, we will use the following known results without mentioning them
again. They follow from the well-understood relation between monadic
second-order logic and automata (\textit{cf.}\ \cite{Thomas-book,PerrinPin}).
\begin{thm}\hfill
  \begin{enumerate}
  \item Let $k\ge2$.
    \begin{itemize}
    \item For any $\omega$-word $\alpha$, there exist finite words $x$
      and $y$ with $xy\equiv_k x$, $yy\equiv_k y$ and $\alpha\equiv_k
      xy^\omega$. Any such pair $(x,y)$ is \emph{a representative of
        the $k$-type of $\alpha$}.
    \item For any $\omega^*$-word $\alpha$, there exist finite words
      $x$ and $y$ with $xy\equiv_k y$, $xx\equiv_k x$ and
      $\alpha\equiv_k x^{\omega^*}y$. Any such pair $(x,y)$ is \emph{a
        representative of the $k$-type of $\alpha$}.
    \item For any bi-infinite word $\xi$, there exist finite words
      $x$, $y$ and $z$ with $xy\equiv_k yz\equiv_k y$, $xx\equiv_k x$,
      $zz\equiv_k z$, and $\xi\equiv_k x^{\omega^*} y z^\omega$. Any
      such triple $(x,y,z)$ is \emph{a representative of the $k$-type
        of $\xi$}.
    \end{itemize}

  \item The following sets are decidable:
    \begin{itemize}
    \item $\{\varphi\in\Sent\mid \forall u\in\{0,1\}^*\colon
      u\models\varphi\}$
    \item $\{(u,\varphi)\mid u\in\{0,1\}^*,\varphi\in\Sent, u\models\varphi\}$
    \item $\{(u,v,\varphi)\mid u,v\in\{0,1\}^*,
      v\neq\varepsilon,\varphi\in\Sent, uv^\omega\models\varphi\}$
    \item $\{(u,v,w,\varphi)\mid u,v,w\in\{0,1\}^*,
      u,w\neq\varepsilon,\varphi\in\Sent, u^{\omega^*}vw^\omega\models\varphi\}$
    \item $\{(u,v,k)\mid u,v\in\{0,1\}^*, k\in\N, u\equiv_k v\}$. This
      means that it is decidable whether $u$ and $v$ represent the
      same $k$-type of finite words.
    \item Similarly, it is decidable whether two pairs of finite words
      represent the same $k$-type of $\omega$-words (of
      $\omega^*$-words, resp). It is also decidable whether two
      triples of finite words represent the same $k$-type of
      bi-infinite words.
    \end{itemize}

  \item $\equiv_k$ is a congruence for concatenation: If
    $u,v\in\{0,1\}^*\cup\{0,1\}^{\omega^*}$ and
    $u',v'\in\{0,1\}^*\cup\{0,1\}^\omega$ with $u\equiv_k v$ and
    $u'\equiv_k v'$, then $uu'\equiv_k vv'$. From representatives of
    the $k$-types of $u$ and $v$, one can compute a representative of
    the $k$-type of $uv$.

  \item $\equiv_k$ is even a congruence for infinite concatenations:
    If $u_i,v_i\in\{0,1\}^+$ with $u_i\equiv_k v_i$ for all $i\in\Z$,
    then the following hold:
    \begin{align*}
      u_0u_1\cdots &\equiv_k v_0v_1\cdots\\
      \cdots u_{-1} u_0 &\equiv_k \cdots v_{-1}v_0\\
      \cdots u_{-1}u_0u_1\cdots &\equiv_k \cdots v_{-1}v_0v_1\cdots
    \end{align*}

  \item If $u\in\{0,1\}^*\cup\{0,1\}^{\omega^*}$ and
    $v\in\{0,1\}^*\cup\{0,1\}^\omega$ such that $\MSOTh(u)$ and
    $\MSOTh(v)$ are both decidable, then $\MSOTh(uv)$ is decidable
    \cite{Shelah}.
  \end{enumerate}
\end{thm}

\subsection{Recursion theoretic notions}

This paper makes use of standard notions in recursion theory; the
reader is referred to \cite{Rog68,Soare} for a thorough
introduction. We assume a canonical effective enumeration
$\Phi_0,\Phi_1,\Phi_2,\ldots$ of all partial recursive functions on
the natural numbers. The set $W_e$ is the domain $\dom(\Phi_e)$ and is
the \emph{$e^{th}$ recursively enumerable set}. Let $\Tot\subseteq\N$ be
the set of natural numbers $e$ such that $\Phi_e$ is total, i.e.,
$W_e=\N$. Furthermore, let $\Rec\subseteq\N$ be the set of natural
numbers $e$ such that $W_e$ is decidable. We will also use the notion
of many-one reductions (or m-reductions) and of Turing-reductions:
$A\le_m B$ denotes the existence of an m-reduction of $A$ to $B$,
$A\le_T B$ that of a Turing-reduction. These relations are transitive
and reflexive, the induced equivalence relations are denoted
$\equiv_m$ and $\equiv_T$, respectively. An equivalence class of
$\equiv_m$ is an \emph{m-degree}, an equivalence class of $\equiv_T$
is a \emph{Turing-degree}. Recall that the class of Turing-degrees,
ordered by $\mathord{\le_T}/\mathord{\equiv_T}$ is an upper
semilattice. For two sets of integers $A$ and $B$, the supremum of
their Turing-degrees is the Turing-degree of $A\oplus B=\{2a\mid a\in
A\}\cup\{2b+1\mid b\in B\}$.

We will, when appropriate, understand an $\omega$-word $\alpha$ as the
set $\alpha^{-1}(1)$. Then, it makes sense to say ``$\alpha$ is
recursive'' meaning ``$\alpha^{-1}(1)$ is recursive'' or to speak
about the degree of $\alpha$. Similarly, an $\omega^*$-word $\alpha$
is identified with the set $-\alpha^{-1}(1)$ and a bi-infinite word
$\xi$ with the set $\{2i\mid i\ge 0,\xi(i)=1\}\cup\{2i+1\mid
i>0,\xi(-i)=1\}$. We also assume an effective enumeration of all
finite words, so that sets of finite words can be understood as
subsets of $\N$, hence the notions ``degree of a set of finite words''
and $\xi\oplus F$ with $\xi\in\{0,1\}^\Z$ and $F\subseteq\{0,1\}^*$
make sense.

A set $A\subseteq\N$ belongs to $\Pi_2$ (the second universal level of
the arithmetical hierarchy) if there exists a decidable set
$P\subseteq\N\times\N^m\times \N^n$ such that $A$ is the set of
natural numbers $a$ satisfying
\[
  \forall x_1,\dots,x_m\exists y_1,\dots y_n\colon P(a,\bar x,\bar y)\,.
\]
A set $B\subseteq\N$ is \emph{$\Pi_2$-hard} if, for every $A\in\Pi_2$,
there exists an m-reduction from $A$ to $B$; the set $B$ is
\emph{$\Pi_2$-complete} if, in addition, $B\in\Pi_2$.

Similarly, $A\subseteq\N$ belongs to $\Sigma_3$ (the third existential
level of the arithmetical hierarchy) if there exists a decidable set
$P\subseteq\N\times\N^\ell\times\N^m\times \N^n$ such that $A$ is the
set of natural numbers $a$ satisfying
\[
  \exists x_1,\dots,x_\ell\forall y_1,\dots,y_m\exists z_1,\dots z_n\colon
  P(a,\bar x,\bar y,\bar z)\,.
\]
The notions $\Sigma_3$-hard and $\Sigma_3$-complete are defined
similarly to the corresponding notions for $\Pi_2$. For our purposes,
it is important that the set $\Tot$ is $\Pi_2$-complete and the set
$\Rec$ is $\Sigma_3$-complete~\cite{Soare}.

\section{When is the $\MSO$-theory of an $\omega$-word decidable?}
\label{sec:omega-decidable}

In this section, we recall the answers by Semenov \cite{Semenov} and
by Rabinovich and Thomas \cite{RabThomas} and we present a slight
strengthening of Semenov's answer.

\subsection{Semenov's characterization}

The first characterisation is provided by Semenov in
\cite{Semenov}. He observed that, in order to decide the MSO-theory of
some $\omega$-word $\alpha$, it is necessary and sufficient to
determine whether words from a given regular set recur in $\alpha$
(and, if not, from what point on no factor of $\alpha$ belongs to the
regular set). This led to the definition of an ``indicator of
recurrence'':

\begin{defi}
  Let $\alpha$ be some $\omega$-word. An \emph{indicator of
    recurrence} for $\alpha$ is a function
  $\indrec\colon\Sent\to\N\cup\{\top\}$ such that, for every $\MSO$-sentence
  $\varphi$, the following hold:
  \begin{itemize}
  \item if $\indrec(\varphi)=\top$, then $\forall k\,\exists j\ge i\ge
    k\colon \alpha[i,j]\models\varphi$
  \item if $\indrec(\varphi)\neq\top$, then $\forall j\ge i\ge
    \indrec(\varphi)\colon \alpha[i,j]\models\neg\varphi$
  \end{itemize}
\end{defi}

Formally, Semenov's formulation from \cite{Semenov} uses the class of
regular languages instead of $\Sent$. He actually means any effective
representation, \textit{e.g.}, as regular expressions or finite automata. Here,
we use the effective representation by $\MSO$-sentences as logic is
the main focus of this paper.

\begin{thm}[Semenov's Characterisation
  \cite{Semenov}]\label{thm:semenov}
  Let $\alpha$ be an $\omega$-word. Then $\MSOTh(\alpha)$ is decidable
  if and only if there is a recursive indicator of recurrence for
  $\alpha$ and the $\omega$-word $\alpha$ is recursive.
\end{thm}

Note that an $\omega$-word can have many recursive indicators of
recurrence: if $\indrec$ is such an indicator, then also
$\varphi\mapsto 2\cdot\indrec(\varphi)$ (with $\top=2\cdot\top$) is an
indicator of recurrence. Differently, there is only one weak indicator
of recurrence of any $\omega$-word $\alpha$:

\begin{defi}
  Let $\alpha$ be some $\omega$-word. The \emph{weak indicator of
    recurrence} for $\alpha$ is the function
  $\windrec\colon\Sent\to\{0,1,\top\}$ defined as follows:
  \[
     \windrec(\varphi)=
     \begin{cases}
       0 & \text{ no factor of $\alpha$ satisfies }\varphi\\
       \top & \text{ there are infinitely many $i\in\N$}\\
       &\qquad\text{such that there exists $j\ge i$ with $\alpha[i,j]\models\varphi$}\\
      1 & \text{ otherwise}
     \end{cases}
  \]
\end{defi}

Note that $\windrec(\varphi)=1$ does not imply that there are only
finitely many factors of $\alpha$ satisfying $\varphi$: let
$\alpha=10^\omega$ and $\varphi=\exists x\colon P(x)$. Then $\varphi$
is satisfied by all factors of the form $\alpha[0,j]$, but by no
factor $w[i,j]$ with $i>0$. Hence $\windrec(\varphi)=1$ in this case.

\begin{cor}\label{cor:semenov+}
  Let $\alpha$ be an $\omega$-word. Then $\MSOTh(\alpha)$ is decidable
  if and only if the weak indicator of recurrence of $\alpha$ is
  recursive and the $\omega$-word $\alpha$ is recursive as well.
\end{cor}

\begin{proof}
  Suppose that $\MSOTh(\alpha)$ is decidable such that, by
  Theorem~\ref{thm:semenov}, $\alpha$ is recursive and there exists a
  recursive indicator of recurrence $\indrec$. Let $\windrec$ be the
  weak indicator of recurrence for $\alpha$. For $\varphi\in\Sent$,
  consider the sentence
  \[
     \psi_\varphi=\exists x, y\colon (x\le y\land\varphi_{x,y})
  \]
  (here $\varphi_{x,y}$ results from $\varphi$ by restricting all
  quantifiers to the interval $[x,y]$).  Then we have
  \[
    \windrec(\varphi)=
    \begin{cases}
      0 & \text{ if }\alpha\models\neg\psi_\varphi\\
      1 & \text{ if }\alpha\models\psi_\varphi\text{ and }\indrec(\varphi)\in\N\\
      \top & \text{ otherwise.}
    \end{cases}
  \]
  Since validity of $\psi_\varphi$ in $\alpha$ is decidable, the
  function $\windrec$ is recursive.

  For the other direction, suppose $\alpha$ is recursive and
  $\windrec$ is the recursive weak indicator of recurrence. We
  construct a recursive indicator of recurrence as follows: If
  $\windrec(\varphi)=\top$, then set $\indrec(\varphi)=\top$. Now
  suppose $\windrec(\varphi)\in\{0,1\}$. For $n\in\N$, consider the sentence
  \[
    \psi_n=\exists x,y\colon (n\le x\le y\land \varphi_{x,y})\,.
  \]
  Since $\windrec(\varphi)\neq\top$, there is a minimal natural number
  $n$ with $\alpha\models\neg\psi_n$. Setting $\indrec(\varphi)=n$
  ensures that $\indrec\colon\Sent\to\N\cup\{\top\}$ is an indicator
  of recurrence. Note that $\indrec(\varphi)$ is minimal among all
  those numbers $n$ satisfying $\windrec(\psi_n)=0$. Hence the
  function $\indrec$ is recursive implying, by
  Theorem~\ref{thm:semenov}, that $\MSOTh(\alpha)$ is decidable.
\end{proof}

\subsection{Rabinovich and Thomas' characterization}

Two other characterisations are given by Rabinovich and Thomas in
\cite{RabThomas}. The idea is to decompose an $\omega$-word into
infinitely many finite sections all of which (except possibly the
first one) have the same $k$-type.

\begin{defi}
  Let $\alpha\in\{0,1\}^\omega$, $u,v\in\{0,1\}^+$, $k\in\N$, and
  $H\subseteq\N$ be infinite.
  \begin{itemize}
  \item The set $H$ is a \emph{$k$-homogeneous factorisation of
      $\alpha$ into $(u,v)$} if $\alpha[0,i-1]\equiv_k u$ and
    $\alpha[i,j-1]\equiv_k v$ for all $i,j\in H$ with $i<j$.
  \item The set $H$ is \emph{$k$-homogeneous for $\alpha$} if it is
    a $k$-homogeneous factorisation of $\alpha$ into some finite words
    $(u,v)$.
  \item Let $H=\{h_i\mid i\in\N\}$ with $h_0<h_1<\dots$. The set $H$
    is \emph{uniformly homogeneous for $\alpha$} if, for all $k\in\N$,
    the set $\{h_i\mid i\ge k\}$ is $k$-homogeneous for $\alpha$.
  \end{itemize}
\end{defi}

As with indicators of recurrence, any $\omega$-word has many uniformly
homogeneous sets: the existence of at least one follows by a repeated
and standard application of Ramsey's theorem, and there are infinitely
many since any infinite subset of a uniformly homogeneous set is again
uniformly homogeneous.

\begin{thm}[1st Rabinovich-Thomas' Characterisation
  \cite{RabThomas}]\label{thm:RabinThomas}
  Let $\alpha$ be an $\omega$-word. Then $\MSOTh(\alpha)$ is decidable
  if and only if there exists a recursive uniformly homogeneous set
  for $\alpha$ and the $\omega$-word $\alpha$ is recursive.
\end{thm}

The idea of the second characterisation by Rabinovich and Thomas is to
compute, from $k\in\N$, a representative of the $k$-type of the
$\omega$-word $\alpha$. This is formalised as follows:

\begin{defi}
  Let $\alpha$ be some $\omega$-word and
  $\typefun\colon\N\to\{0,1\}^+\times\{0,1\}^+$. The function
  $\typefun$ is a \emph{type-function for $\alpha$} if, for all
  $k\in\N$, the $\omega$-word $\alpha$ has a $k$-homogeneous
  factorisation into $\typefun(k)=(u,v)$.
\end{defi}

Let $\typefun$ be a type-function for the $\omega$-word $\alpha$ and
let $k\in\N$. Then there exists a $k$-homogeneous factorisation $H$ of
$\alpha$ into $\typefun(k)=(u,v)$. Let $H=\{h_i\mid i\in\N\}$ such
that $h_0<h_1<\dots$. Then we get
\begin{align*}
  \alpha &=
   \alpha[0,h_0-1]\,\alpha[h_0,h_1-1]\,\alpha[h_1,h_2-1]\cdots\\
    &\equiv_k u v^\omega\,.
\end{align*}
Furthermore,
$v\equiv_k\alpha[h_0,h_2-1]=\alpha[h_0,h_1-1]\,\alpha[h_1,h_2-1]\equiv_k
vv$. Consequently, $\typefun(k)$ is a representative of the $k$-type
of $\alpha$. Recall that validity in $\alpha$ of a sentence of
quantifier depth $k$ can be determined from any representative of the
$k$-type of $\alpha$. Hence, to decide the MSO-theory of $\alpha$, it
suffices to have a recursive type-function. The converse implication
of the following theorem holds since the ``minimal type function'' can
be expressed in MSO (\textit{cf.}~proof of Theorem~\ref{thm:RabinThomas2-z}).

\begin{thm}[2nd Rabinovich-Thomas' Characterisation
  \cite{RabThomas}]\label{thm:RabinThomas2}
  Let $\alpha$ be an $\omega$-word. Then $\MSOTh(\alpha)$ is decidable
  if and only if $\alpha$ has a recursive type-function.
\end{thm}

Note that, differently from Theorem~\ref{thm:RabinThomas}, this theorem
does not mention that $\alpha$ is recursive. But this recursiveness is
implicit: Let $\typefun$ be a recursive type-function and
$k\in\N$. Then one can write down a first-order sentence of
quantifier-depth $k+2$ expressing that $\alpha(k)=1$. Let
$\typefun(k+2)=(u,v)$. Then $\alpha\equiv_{k+2}uv^\omega$ implies
$\alpha(k)=uv^k(k)$, hence $\alpha(k)$ is computable from $k$, i.e.,
$\alpha^{-1}(1)$ is recursive.

\section{How difficult is it to tell whether the $\MSO$-theory of an
  $\omega$-word is decidable?}
\label{sec:complexity}

In this section we show that the question whether $\MSOTh(\alpha)$ is
decidable for a recursive $\omega$-word $\alpha$ is
$\Sigma_3$-complete.

Technically, we will consider the following two sets:
\begin{align*}
    \DecTh^{\MSO}_\N &= \{e\in\Rec\mid
            \MSOTh(\N,\le,W_e) \text{ is decidable}\}\\
    \UndecTh^{\MSO}_\N &= \{e\in\Rec\mid
            \MSOTh(\N,\le,W_e) \text{ is undecidable}\}
\end{align*}
Note that $(\N,\le,W_e)$ is the labeled linear order $M_w$ associated
to the characteristic $\omega$-word $\alpha$ of the $e^{th}$
recursively enumerable set $W_e$. We will prove that the first set is
in $\Sigma_3$ and that any separator of the two sets (i.e., any set
containing $\DecTh^{\MSO}_\N$ and disjoint from $\UndecTh^{\MSO}_\N$)
is $\Sigma_3$-hard.

\begin{lem}\label{lem:Sigma3-upperbound}
  The set $\DecTh^{\MSO}_\N$ belongs to $\Sigma_3$.
\end{lem}

We present two proofs of this lemma, one based on the first
Rabinovich-Thomas characterisation, the second one based on the
Semenov characterization.

\begin{proof} (based on Theorem~\ref{thm:RabinThomas})
  Let $\alpha$ be some recursive $\omega$-word.

  Recall that a set $H\subseteq\N$ is infinite and recursive iff there
  exists a total computable and strictly monotone function $f$ such
  that $H=\{f(n)\mid n\in\N\}$. Now consider the following statement:
  \[
    \exists e\,\forall k,i,j,i',j'\colon
    \begin{array}[t]{ll}
      e\in\Tot\land{\ } \\
      i<j\Rightarrow \Phi_e(i)< \Phi_e(j)\land{\ }\\
      (k\le i\le j\land k\le i'\le j'\Rightarrow
      \alpha[\Phi_e(i),\Phi_e(j)]\equiv_k \alpha[\Phi_e(i'),\Phi_e(j')]
    \end{array}
  \]
  It expresses that there exists a total computable function (namely
  $\Phi_e$) that is strictly monotone. Its image then consists of the
  numbers
  \[
    \Phi_e(0)<\Phi_e(1)<\Phi_e(2)<\dots\,.
  \]
  The last line expresses that this image is uniformly homogeneous for
  $\alpha$. Hence this statement says that there exists a recursive
  uniformly homogeneous set for $\alpha$, i.e., that $\MSOTh(\alpha)$
  is decidable by Theorem~\ref{thm:RabinThomas}.

  Let $k,i,i',j,j'\in\N$ with $k\le i\le j$ and $k\le i'\le j'$. Then
  we can compute the finite words $\alpha[\Phi_e(i),\Phi_e(j)]$ and
  $\alpha[\Phi_e(i'),\Phi_e(j')]$ since $\alpha$ is recursive. Hence
  it is decidable whether
  \[
    \alpha[\Phi_e(i),\Phi_e(j)]\equiv_k \alpha[\Phi_e(i'),\Phi_e(j')]
  \]
  holds.

  Since $\Tot\in\Pi_2$, the whole statement is consequently in
  $\Sigma_3$.
\end{proof}

\begin{proof} (based on Theorem~\ref{thm:semenov}) We enumerate the set
  $\Sent$ of $\MSO$-sentences in any effective way as
  $\varphi_0,\varphi_1,\dots$. Let $e\in\Tot$ and consider the
  function $\indrec\colon\Sent\to\N\cup\{\top\}$ defined by
  \[
    \varphi_i\mapsto
    \begin{cases}
     \Phi_e(i)-1 & \text{ if }\Phi_e(i)>0\\
     \top & \text{ if }\Phi_e(i)=0\,.
    \end{cases}
  \]
  This function is an indicator of recurrence for the $\omega$-word
  $\alpha$ if and only if the following holds:
  \[
  \forall \varphi\in\Sent\colon
  \begin{array}[t]{ll}
    \indrec(\varphi)\neq\top\Rightarrow\forall k\ge j\ge \indrec(\varphi)\colon \alpha[j,k]\models\lnot\varphi
    \land{\ }\\
    \indrec(\varphi)=\top\Rightarrow\forall j\exists \ell\ge k\ge j\colon \alpha[k,\ell]
       \models\varphi
  \end{array}
  \]
  Given the definition of $\indrec$, this is equivalent to saying
  \[
  \forall i\colon
  \begin{array}[t]{ll}
    \Phi_e(i)>0\Rightarrow\forall k\ge j\ge \Phi_e(i)\colon \alpha[j,k]\models\lnot\varphi_i
    \land{\ }\\
    \Phi_e(i)=0\Rightarrow\forall j\exists \ell\ge k\ge j\colon \alpha[k,\ell]
       \models\varphi_i\,.
  \end{array}
  \]
  If $\alpha$ is recursive, this is a $\Pi_2$-statement. Consequently, also
  the existence of a recursive indicator of recurrence is a
  $\Sigma_3$-statement.
\end{proof}

We could present a third proof based on Corollary~\ref{cor:semenov+}: For
any $\varphi\in \Sent$, let $\varphi'_n$ denote the sentence $\exists
y\ge x\ge n\colon \varphi_{x,y}$ where $\varphi_{x,y}$ results from
$\varphi$ by restricting all quantifiers to the interval $[x,y]$. Then
the proof can be constructed in the same way as the proof above,
except using the following $\Pi_2$ statement:
\[\forall \varphi\in\Sent\colon
  \begin{array}[t]{ll}
    \indrec(\varphi)=0\Rightarrow \forall k\ge j\ge 0\colon \alpha[j,k]\models \lnot\varphi \land {\ }\\
    \indrec(\varphi)=\top\Rightarrow\forall j\exists \ell\ge k\ge j\colon \alpha[k,\ell]
       \models\varphi \land{\ }\\
\indrec(\varphi)\notin\{0,\top\}\Rightarrow\indrec(\varphi)=1
  \end{array}
\]

\begin{rem}
  From the second characterisation by Rabinovich and Thomas
  (Theorem~\ref{thm:RabinThomas2}), we can only infer that
  $\DecTh^{\MSO}_\N$ is in $\Sigma_5$:

  Let $\alpha$ be some recursive $\omega$-word and
  $u,v\in\{0,1\}^+$. Then, by the proof of \cite[Proposition~7]{RabThomas},
  there exists a $k$-homogeneous factorisation of $\alpha$ into
  $(u,v)$, if the following $\Sigma_3$-statement $\varphi(u,v)$ holds:
  \begin{align*}
    \exists x\forall y\exists z,z'\colon(\alpha[0,x-1]\equiv_k u\land
     y<z<z'\land \alpha[x,z-1]\equiv_k\alpha[z,z'-1]\equiv_k v)
  \end{align*}
  Hence a function $\typefun\colon\N\to\{0,1\}^+\times\{0,1\}^+$ is a
  type-function for $\alpha$ iff the $\Pi_4$-statement $\forall
  k\in\N\colon\varphi(\typefun(k))$ holds. Consequently, there is a
  recursive type-function iff we have
  \[
    \exists e\colon e\in\Tot\land\forall k\colon\varphi(\mathrm{pair}(\Phi_e(k)))
  \]
  where $\mathrm{pair}\colon\N\to \{0,1\}^+\times\{0,1\}^+$ is a
  computable surjection. Since this statement is an
  $\Sigma_5$-statement, the claim follows.
\end{rem}

\begin{rem}
  Recall that any MSO-sentence can be translated into a deterministic
  parity automaton that accepts precisely those words that satisfy the
  sentence (\textit{cf.}~\cite{PerrinPin}). Hence, $\MSOTh(\alpha)$ is
  decidable if and only if the set of deterministic parity automata
  accepting $\alpha$ is decidable. This statement is a
  $\Sigma_4$-statement.
\end{rem}

Three (out of five) characterisations of the decidable recursive
$\omega$-words result in the same recursion-theoretic upper bound
$\Sigma_3$ of the set $\DecTh^{\MSO}_\N$. It is therefore natural to
ask if these characterisations are ``optimal''. Namely, if one can
separate $\DecTh^{\MSO}_\N$ from $\UndecTh^{\MSO}_\N$ using a simpler
statement. We now prepare a negative answer to this question (which is
an affirmative answer to the optimality question posed first).

\begin{lem}\label{lem:ell_k}
  From $k\in\N$, one can compute $\ell\in\N$ such that
  $0^\ell\equiv_k 0^{2\ell}$.
\end{lem}

\begin{proof}
  Up to logical equivalence, there are only finitely many
  $\MSO$-sentences of quantifier-rank at most $k$. Hence there are
  only finitely many $\equiv_k$-equivalence
  classes. Consequently, there are $i,j\ge1$ with $0^i\equiv_k
  0^{i+j}$. Even more, we can effectively find such a pair by simply
  checking all pairs $(i,j)$ (since $k$-equivalence of finite words is
  decidable).

  With $\ell=ij$, we then get
  \[
    0^\ell= 0^{i}0^{\ell-i} \equiv_k 0^{i+ij} 0^{\ell-i} = 0^{2\ell}
  \]
  where $0^{i}0^{\ell-i} \equiv_k 0^{i+ij}0^{\ell-i}$ follows
  from $0^i\equiv_k 0^{i+j}$.
\end{proof}

We now construct an m-reduction from $\Rec$ to any separator of the sets
$\DecTh^{\MSO}_\N$ and $\UndecTh^{\MSO}_\N$: Let $e\in\N$. Then the
sets $\{2a\mid a\in W_e\}$ and $2\N+1$ are both (effectively)
recursively enumerable and so is their union. Hence, by
\cite[Corollary~5.V(d)(i)]{Rog68}, one can compute $f\in\N$ such that
$\Phi_f$ is total and injective and
\[
  \{2a\mid a\in W_e\}\cup(2\N+1)=\{\Phi_f(i)\mid i\in\N\}\,.
\]
For $i\in\N$, set
\begin{equation}\label{eq:c-i}
  x_i =
  2^{\Phi_f(i)} \cdot \prod_{0\leq j\leq i} (2j+1)
\end{equation}
and consider the $\omega$-word $\alpha_e=1 0^{x_0} 1 0^{x_1} 1
0^{x_2}\cdots$.  Since $\Phi_f$ is total, this $\omega$-word is
recursive.

\begin{lem}\label{lem:Sigma3-MSO}
  Let $e\in\N$. The $\MSO$-theory of the $\omega$-word $\alpha_e$ is
  decidable if and only if the $e^{th}$ recursively enumerable set
  $W_e$ is recursive, i.e., $e\in\Rec$.
\end{lem}

\begin{proof}
  First suppose that the $\MSO$-theory of $\alpha_e$ is decidable.
  For $a\in\N$, we have $a\in W_e$ iff there exists $i\ge 0$ with
  $2a=\Phi_f(i)$ iff there exists $i\ge 0$ such that $2^{2a}$ is the
  greatest power of $2$ that divides $x_i$. Consequently, $a\in W_e$
  iff the $\omega$-word $\alpha_e$ satisfies
  \begin{equation}
    \label{eqn:FO+MOD}
    \exists x,y\in P\colon
    \begin{array}[t]{ll}
      x<y\land\forall z\colon (x<z<y\Rightarrow z\notin P) \land\\
      2^{2a}\mid y-x-1\land 2^{2a+1}\not|\; y-x-1
    \end{array}
  \end{equation}
  By Example~\ref{exa-2.1}, $n\mid y-x-1$ is expressible by an
  $\MSO$-formula, i.e., the above formula can be written as an
  $\MSO$-sentence. Since validity in $\alpha_e$ of the resulting
  $\MSO$-sentence is decidable, the set $W_e$ is recursive.

  Conversely, let $W_e$ be recursive. To show that the $\MSO$-theory
  of $\alpha_e$ is decidable, let $\varphi$ be some
  $\MSO$-sentence. Let $k=\qr(\varphi)$ be the quantifier-rank of
  $\varphi$. To decide whether $\alpha_e\models\varphi$, we proceed as
  follows:
  \begin{itemize}
  \item First, compute $\ell>0$ such that $0^\ell\equiv_k
    0^{2\ell}$. This is possible by Lemma~\ref{lem:ell_k}.

  \item Next determine $a,b\in\N$ such that $\ell=2^a(2b+1)$.

  \item Then compute $i\ge b$ such that $\Phi_f(j)>a$ for all $j>i$:
    to this aim, first determine $A=\{n\le a\mid n\in W_e\text{ or $a$
      odd}\}$ which is possible since $W_e$ is decidable. Then compute
    the least $i\ge b$ such that $A\subseteq\{\Phi_f(j)\mid j\le
    i\}$. Since $\Phi_f$ is injective, we get $\Phi_f(j)>a$ for all
    $j>i$.

  \item Decide whether $10^{x_0} 1 0^{x_1}\dots 1 0^{x_i} (1
    0^\ell)^\omega$ satisfies $\varphi$ which is possible since this
    $\omega$-word is ultimately periodic.
  \end{itemize}
  Let $j>i$. Then $\Phi_f(j)>a$ and $j>i\ge b$ imply that $x_j$ is a
  multiple of $\ell$. Consequently $0^{x_j}\equiv_k 0^\ell$. We
  therefore obtain
  \[
     \alpha_e \equiv_k 10^{x_1}10^{x_2}\cdots 10^{x_i}(10^\ell)^\omega\,.
  \]
  Hence the above algorithm is correct.
\end{proof}

Since $\Rec$ is $\Sigma_3$-complete \cite[Theorem 14.XVI]{Rog68},
Lemma~\ref{lem:Sigma3-MSO} and Lemma~\ref{lem:Sigma3-upperbound} imply
that the problem of deciding whether a recursive $\omega$-word has a
decidable $\MSO$-theory is $\Sigma_3$-complete:

\begin{thm}\label{thm:Sigma3-MSO}\hfill
  \begin{itemize}
  \item $\DecTh^{\MSO}_\N$ is in $\Sigma_3$.
  \item Any set containing $\DecTh^{\MSO}_\N$ and disjoint from
    $\UndecTh^{\MSO}_\N$ is $\Sigma_3$-hard.
  \end{itemize}
\end{thm}

\noindent\textbf{Remark.}
Theorem~\ref{thm:RabinThomas} is proved in \cite{RabThomas} not only
for the logic $\MSO$, but also for the weaker logic $\FO$ and for the
intermediate logic $\FOMOD$ that extends $\FO$ by modulo-counting
quantifiers. Consequently, Lemma~\ref{lem:Sigma3-upperbound} also
holds for the logics $\FO$ and $\FOMOD$ \textit{mutatis mutandis}.

On the other hand, (\ref{eqn:FO+MOD}) in the proof of
Lemma~\ref{lem:Sigma3-MSO} can easily be expressed in $\FOMOD$
implying that also Lemma~\ref{lem:Sigma3-MSO} holds for this logic.
Furthermore, one may use a very similar reduction to prove the same
$\Sigma_3$-bound for the $\FO$-theory: replace the definition of $x_i$
from \eqref{eq:c-i} by $x_i=\Phi_f(j)$ (and $0^\ell\equiv_k 0^{2\ell}$
by $0^\ell\equiv_k^{\FO}0^{\ell+1}$ in Lemma~\ref{lem:ell_k}). A
similar argument as in Lemma~\ref{lem:Sigma3-MSO} proves that $W_e$ is
recursive if and only if the $\omega$-word $\alpha_e$ obtained this
way has a decidable $\FO$-theory.

Thus, the above Theorem~\ref{thm:Sigma3-MSO} also holds for the logics
$\FO$ and $\FOMOD$.

\section{When is the $\MSO$-theory of a bi-infinite word decidable?}
\label{sec:bi-infinite}
In this section, we investigate whether the characterisations from
Theorems~\ref{thm:semenov}, \ref{thm:RabinThomas}, and
\ref{thm:RabinThomas2} and from Corollary~\ref{cor:semenov+} can be
lifted from $\omega$- to bi-infinite words.

A crucial notion will be that of the theory of a language: Let
$L\subseteq\{0,1\}^*$ be a language. Its \emph{$\MSO$-theory}
$\MSOTh(L)$ is the set of sentences $\varphi\in\Sent$ such that
$w\models\varphi$ for all $w\in L$, i.e.,
$\MSOTh(L)=\bigcap_{w\in L}\MSOTh(w)$.

In \cite[pages 602-603]{Semenov}, Semenov proves the following
characterizations:
\begin{thm}[\cite{Semenov}]\label{thm:Semenov}
  Let $\xi$ be a bi-infinite word.
  \begin{enumerate}
  \item If $\xi$ is not recurrent, then $\MSOTh(\xi)$ is decidable if
    and only if $\MSOTh(\xi(-\infty,-1])$ and $\MSOTh(\xi[0,\infty))$
    are both decidable.
  \item If $\xi$ is recurrent, then $\MSOTh(\xi)$ and $\MSOTh(F(\xi))$
    are interreducible. In particular, in this case $\MSOTh(\xi)$ is
    decidable if and only if $\MSOTh(F(\xi))$ is decidable.
  \end{enumerate}
\end{thm}

\subsection{Complicated bi-infinite words with decidable $\MSO$-theory}

We first demonstrate that, in the second statement of
Theorem~\ref{thm:Semenov}, we cannot replace the decidability of
$\MSOTh(F(\xi))$ by that of $F(\xi)$.

\begin{lem}\label{lem:factors}
  Let $L\subseteq\{0,1\}^*$ be a set of finite words. Then the
  following are equivalent:
  \begin{itemize}
  \item There exists a recurrent bi-infinite word $\xi$ with $F(\xi)=L$.
  \item
  \begin{enumerate}[label=(\alph*)]
  \item $L$ contains a non-empty word.
  \item If $uvw\in L$, then $v\in L$.
  \item For any $u,w\in L$, there is a finite word $v$ such that
    $uvw\in L$
  \end{enumerate}
  \end{itemize}
  In addition, $\xi$ can be chosen recursive iff $L$ is recursively
  enumerable.
\end{lem}

\begin{proof}
  First suppose $L=F(\xi)$ for some recurrent bi-infinite word
  $\xi$. Then (a), (b), and (c) are obvious. If, in addition, $\xi$ is
  recursive, then its set of factors $L$ is recursively enumerable.

  Conversely, suppose (a), (b), and (c) hold and $L$ is recursively
  enumerable (the proof for non-recursively enumerable sets $L$ can be
  extracted easily from this one). By (a), there exists a non-empty
  word $u\in L$. From (c), we obtain that $L$ is infinite. Let
  $f\colon\N\to\{0,1\}^*$ be a computable and total function with
  $L=\{f(i)\mid i\in\N\}$. We will write $u_i$ for the word $f(i)$.
  Inductively, we construct two sequences $(x_i)_{i>0}$ and
  $(y_i)_{i>0}$ of words from $L$ such that, for all $i\in\N$, the
  finite word
  \[
     w_i=u_i x_i u_{i-1} x_{i-1}\dots u_1 x_1 u_0 y_1 u_1 y_2 u_2 \dots
  y_i u_i
  \]
  belongs to $L$.

  Let $i>0$ and suppose we already defined the words $x_j$ and $y_j$
  for $j<i$ such that $w_{i-1}\in L$. To extend $w_{i-1}$ to the left,
  let $j\in\N$ be the minimal index with $f(j)\in u_i\{0,1\}^*
  w_{i-1}$ (such a number $j$ exists by (b)). Choose $x_i\in\{0,1\}^*$
  with $f(j)=u_i x_i w_{i-1}$. Next we extend this word from $L$
  symmetrically to the right: let $k\in\N$ be minimal with $f(k)\in
  u_i x_i w_{i-1} \{0,1\}^* u_i$ and choose $y_i\in\{0,1\}^*$ such that
  $f(k)= u_i x_i w_{i-1} y_i u_i$.

  Then the bi-infinite word
  \[
     \xi=\cdots u_3 x_3\, u_2 x_2\, u_1 x_1\, u_0\, y_1 u_1\, y_2 u_2 \,
         y_3 u_3\cdots
  \]
  satisfies $L\subseteq F(\xi)$.

  Let $v\in\{0,1\}^*$ be some factor of $\xi$. Then there is $i\in\N$
  such that $v$ is a factor of $w_i$. Since $w_i\in L$, condition (b)
  implies $v\in L$. Hence $F(\xi)=L$.

  Now let $v\in F(\xi)=L$. By (c), there are infinitely many $i\in\N$
  such that $v$ is a factor of $u_i$. Hence $\xi$ is recurrent. It is
  also recursive since the word $w_i$ is computable from $w_{i-1}$.
\end{proof}

\begin{thm}
  There exists a recurrent and recursive bi-infinite word $\xi$ whose
  set of factors is decidable, but $\MSOTh(\xi)$ is undecidable.
\end{thm}

\begin{proof}
  Let $f\colon\N\to\N$ be some recursive and total function such that
  $\{f(i)\mid i\in\N\}$ is not recursive. Let $L\subseteq\{0,1\}^*$ be
  the set of all finite words $u$ with the following property: If
  $10^{2i+1}10^{2j}1$ is a factor of $u$, then $j=f(i)$. This set is
  clearly recursive, contains a non-empty word, and satisfies
  conditions (a), (b), and (c) from Lemma~\ref{lem:factors}. Hence
  there exists a recurrent and recursive bi-infinite word $\xi$ with
  $F(\xi)=L$.

  For $j\in\N$, consider the following sentence:
  \[
    \exists x<y\colon
    \begin{array}[t]{ll}
      & P(x)\land P(y+2j) \land \lnot\, 2\mid y-x-1\\
      \land & \forall z\colon (x<z<y+2j\land P(z)\to z=y)
    \end{array}
  \]
  It expresses that the language $1(00)^*010^{2j}1$ contains a factor
  of $\xi$. But this is the case iff it contains a factor of some word
  from $L$ iff there exists $i\in\N$ with $j=f(i)$. Since this is
  undecidable, the MSO-theory of $\xi$ is undecidable.
\end{proof}

Suppose $\xi$ is not recurrent with decidable $\MSO$-theory. Then by
the first statement of Theorem~\ref{thm:Semenov}, the $\MSO$-theories
of the two ``halves'' of $\xi$ are decidable. Hence these two halves
are recursive implying that $\xi$ is recursive as well.

Our next two theorems show that the situation is ``in some sense more
exotic'' (as Semenov puts it \cite[page 165]{Semenov2}) when we
consider recurrent bi-infinite words. Namely, we construct
non-recursive bi-infinite words with decidable $\MSO$-theories whose
``halves'' have undecidable $\MSO$-theories.

\begin{thm}\label{thm:non-recursive-but-decidable-MSOTh}
  There exists a recursive and recurrent bi-infinite word $\xi$ with
  decidable $\MSO$-theory, such that every nontrivial m-degree
  $\mathbf{a}$ contains some bi-infinite word $\xi_{\mathbf{a}}$ with
  $\MSOTh(\xi)=\MSO(\xi_{\mathbf{a}})$. Furthermore,
  $\xi_{\mathbf{a}}(-\infty,-1]$ and $\xi_{\mathbf{a}}[0,\infty)$ both
  belong to $\mathbf{a}$.
\end{thm}

\begin{proof}
  Let $g\colon\N\to\{0,1\}^*$ be a computable surjection and define
  \[
     \xi=\dots g(2)\,g(1)\, g(0)\, g(1)\, g(2) \dots\,.
  \]
  Clearly, $\xi$ is recursive and recurrent with
  $F(\xi)=\{0,1\}^*$. Since $\MSOTh(\{0,1\}^*)$ is decidable, the
  $\MSO$-theory of the recurrent bi-infinite word $\xi$ is decidable
  by Theorem~\ref{thm:Semenov}.

  Since $\xi$ is recursive, we can set $\xi_{\mathbf{a}}=\xi$ for the
  minimal nontrivial $m$-degree $\mathbf{a}$ of all nontrivial
  recursive sets.

  Now let $\mathbf{a}$ be some $m$-degree above the $m$-degree of all
  nontrivial recursive sets. Furthermore, let $A\in\mathbf{a}$ be an
  arbitrary set in the m-degree $\mathbf{a}$. Since $\mathbf{a}$ is
  nontrivial, we get $\emptyset\neq A\neq\N$. We denote the
  characteristic function of $A$ by $\chi_A$. Then let
  $\beta_{\mathbf{a}}$ be the $\omega$-word
  \[
  \beta_{\mathbf{a}}= \chi_A(0) g(0) \chi_A(1) g(1) \chi_A(2) g(2) \cdots
  \]
  and set $\xi_{\mathbf{a}}=\beta_{\mathbf{a}}^R\,\beta_{\mathbf{a}}$.

  Since $g$ is recursive, we get
  $\beta_{\mathbf{a}}\le_m A$. Conversely, $n\in A$ iff the
  $\omega$-word $\beta_{\mathbf{a}}$ carries $1$ at position
  $\sum_{0\le i<n}(1+|g(i)|)$, i.e., $A\le_m\beta_{\mathbf{a}}$.  This
  proves $\beta_{\mathbf{a}}\in\mathbf{a}$. It follows that also
  $\xi_{\mathbf{a}}\in\mathbf{a}$.

  Note that also the bi-infinite word $\xi_{\mathbf{a}}$ is recurrent
  with $F(\xi_{\mathbf{a}})=\{0,1\}^*$. Hence,  by Theorem~\ref{thm:Semenov},
  $\MSOTh(\xi)=\MSOTh(\xi_{\mathbf{a}})$.
\end{proof}

The above theorem provides us with a theory $\MSOTh(\xi)$ that is
realised (by some bi-infinite word) in every non-trivial $m$-degree
$\mathbf{a}$. We next ask to what extend this holds for every
MSO-theory. First, if $\xi$ is non-recurrent or periodic, then $\xi$
is computable from its MSO-theory. Hence, all realisations of
$\MSOTh(\xi)$ are computable in $\MSOTh(\xi)$ and are therefore of
bounded complexity. It remains to consider the recurrent, non-periodic
case. We therefore first demonstrate some facts about the factor set
$F(\xi)$ of a recurrent bi-infinite word $\xi$.

\begin{defi}
  Let $L\subseteq\{0,1\}^*$ be a language. A word $u\in L$ is
  \emph{left-determining in $L$} if for every $k\in\N$ there is exactly
  one word $vu\in L$ with $|v|=k$. Similarly, $u$ is
  \emph{right-determining in $L$} if for every $k\in \N$ there is exactly
  one word $uv\in L$ with $|v|=k$. The word $u\in L$ is
  \emph{determining in $L$} if it is both left- and right-determining.
\end{defi}

Intuitively a word $w\in L$ is left-determining (right-determining) in
$L$ if it can be extended on the left (right) in a unique way.

\begin{lem}\label{lem:determined}
  Let $\xi$ be a recurrent bi-infinite word. The following are
  equivalent:
  \begin{enumerate}
  \item[(1)] $\xi$ is periodic.
  \item[(2)] $F(\xi)$ contains a determining word.
  \item[(3)] $F(\xi)$ contains a right-determining word.
  \item[(3')] $F(\xi)$ contains a left-determining word.
  \end{enumerate}
\end{lem}

\begin{proof}
  For (1)$\to$(2), let $\xi= u^{\omega^*} u^\omega$ be a periodic
  word. Then $u$ is determining in $F(\xi)$. The direction (2)$\to$(3)
  is trivial by the very definition.

  For (3)$\to$(1), suppose $u$ is a right-determining word in
  $F(\xi)$. Choose $i<j$ such that $\xi[i,i+|u|-1]=\xi[j,j+|u|-1]=u$
  (such a pair $i<j$ exists since $\xi$ is recurrent). With $p=j-i$,
  we claim $\xi(n)=\xi(n+p)$ for all $n\in\Z$: First let $n\ge
  j+|u|$. Then $\xi[i,n]$ and $\xi[j,n+p]$ are two words from $F(\xi)$
  that both start with $u$. We have
  $|\xi[i,n]|=n-i-1=n+p-j-1=|\xi[j,n+p]|$. Since $u$ is
  right-determining, this implies $\xi[i,n]=\xi[j,n+p]$ and therefore
  $\xi(n)=\xi(n+p)$. Consequently,
  $\xi[j+|u|,\infty)=\xi[j+|u|,j+|u|+p]^\omega$. Next let
  $n<j+|u|$. Since $\xi$ is recurrent, there is $k<n$ with
  $\xi[k,k+|u|-1]=u$. Since $u$ is right-determining, this implies
  $\xi[k,\infty)=\xi[j+|u|,\infty)=\xi[j+|u|,j+|u|+p]^\omega$ and
  therefore in particular $\xi(n)=\xi(n+p)$.

  The implications (2)$\to$(3')$\to$(1) are shown analogously.
\end{proof}

Lemma~\ref{lem:determined} states that a recurrent non-periodic
bi-infinite word does not contain any left-determining or
right-determining factor, and thus can be extended in both directions
(left and right) in at least two ways without changing the set of its
factors. This observation allows to prove the following:
\begin{lem}\label{lem:type-A}
  Let $\xi$ be a recurrent non-periodic bi-infinite word and let
  $f_\xi\colon\N\to F(\xi)$ be a surjection (that we identify with the
  relation $\{(n,f_\xi(n))\mid n\in\N\}$). For any set
  $A\subseteq \N$, there is a recurrent bi-infinite word $\xi_A$ such
  that $F(\xi)=F(\xi_A)$ and
  $\xi_A(-\infty,-1]\oplus f_\xi\equiv_T \xi_A[0,\infty)\oplus
  f_\xi\equiv_T A\oplus f_\xi$.
\end{lem}
\newpage

\begin{proof}
  In the following, we write $w_s$ for the factor $f_\xi(s)$ of $\xi$.

  Now let $A\subseteq\N$ be arbitrary. We will construct a sequence of
  tuples
  \[
    t_s=(u_s,v_s,x_s,y_s)\in(\{0,1\}^*)^4
  \]
  such that, for all $s\in\N$, the finite word
  \begin{align*}
     z_s&= w_sy_sv_s\;z_{s-1}\,u_sx_sw_s \qquad
              \text{(with $z_{-1}=\varepsilon$)}\\
        &=w_sy_sv_s\,w_{s-1}y_{s-1}v_{s-1}\dots w_0y_0v_0\;
     u_0x_0w_0 \dots u_{s-1}x_{s-1}w_{s-1}\,u_sx_sw_s
  \end{align*}
  is a factor from $F(\xi)$ (the bi-infinite word $\xi_A$ will be the
  ``limit'' of these words).

  To start with $s=0$ note the following: since $\xi$ is recurrent and
  $w_0\in F(\xi)$, the bi-infinite word $\xi$ contains a factor from
  $w_0 \{0,1\}^* w_0$. Let $n\in\N$ be minimal with
  $f_\xi(n)\in w_0\{0,1\}^*w_0$. Choose $y_0\in\{0,1\}^*$ such that
  $f_\xi(n)=w_0\,y_0\,w_0$ and set $u_0=v_0=x_0=\varepsilon$.

  For the induction step, assume that we constructed the tuple $t_s$
  and that $z_s$ is a factor of $\xi$. Since $\xi$ is recurrent but
  not periodic, the word $z_s$ is not right-determining in $F(\xi)$ by
  Lemma~\ref{lem:determined}. Hence there are two distinct finite
  words $u$ and $u'$ of the same length such that
  $z_s u,z_s u'\in F(\xi)$.  Let $(k,\ell)\in\N^2$ be the
  lexicographically minimal pair with
  $f_\xi(k),f_\xi(\ell)\in z_s\{0,1\}^*$, $|f_\xi(k)|=|f_\xi(\ell)|$, and
  $f_\xi(k)\neq f_\xi(\ell)$. Choose the word $u_{s+1}$ such that
  \[
      z_s u_{s+1}=
      \begin{cases}
        f_\xi(k) & \text{if }s\in A\\
        f_\xi(\ell) & \text{otherwise.}
      \end{cases}
  \]
  Now the word $z_s u_{s+1}$ is a factor of $\xi$. Since $\xi$ is
  recurrent, it has some factor from
  $z_s u_{s+1} \{0,1\}^* w_{s+1}$. Choose $m\in\N$ minimal such that
  $f_\xi(m)$ belongs to this set and let $x_{s+1}\in\{0,1\}^*$ such
  that $f_\xi(m)=z_s u_{s+1} x_{s+1} w_{s+1}$.

  To choose $v_{s+1}$ and $y_{s+1}$, we proceed symmetrically to the
  left: The word $z_s'=z_s u_{s+1} x_{s+1} w_{s+1}$ is a factor of
  $\xi$ that is not left-determining. Hence there exists a pair of
  distinct words $v$ and $v'$ of the same length with
  $v z_s',v' z_s'\in F(\xi)$. Choose the pair $(k',\ell')\in\N^2$
  lexicographically minimal with
  $f_\xi(k'),f_\xi(\ell')\in\{0,1\}^* z_s'$ of the same length but
  distinct. Then choose the word $v_{s+1}$ such that
  \[
    v_{s+1} z_s'=
    \begin{cases}
      f_\xi(k') & \text{if }s\in A\\
      f_\xi(\ell') & \text{otherwise.}
    \end{cases}
  \]
  Since the word $v_{s+1} z_s'$ is a factor of the recurrent word
  $\xi$, we can choose $m'\in\N$ minimal with
  $f_\xi(m)\in w_{s+1}\{0,1\}^* v_{s+1} z_s'$. Then let
  $y_{s+1}\in\{0,1\}^*$ such that  $f_\xi(m)= w_{s+1} y_{s+1} v_{s+1} z_s'$.
  This completes the construction of the tuple $t_{s+1}$ and therefore
  the inductive construction of all the tuples $t_s$.

  Now set $\xi_A=\cdots w_1y_1v_1\,w_0y_0v_0\;
     u_0x_0w_0\,u_1x_1w_1 \cdots\,$.

  Observe the following:
  \begin{itemize}
  \item Note that $F(\xi)=\{w_s\mid s\in\N\}\subseteq F(\xi_A)$.
  \item Let $u\in F(\xi_A)$. There exists $s\in\N$ such that
    $u\in F(z_s)$ implying $F(\xi_A)\subseteq F(\xi)$ since
    $z_s\in F(\xi)$.
  \end{itemize}
  Hence $F(\xi)=F(\xi_A)$. To show that $\xi_A$ is recurrent, let
  $u\in F(\xi_A)$. Then there are infinitely many $s\in \N$ with
  $u\in F(w_s)$. Hence $u$ appears in $\xi_A$ before and beyond every
  position, i.e., $\xi_A$ is indeed recurrent.

  Since the above describes how to compute the bi-infinite word
  $\xi_A$ using the oracles $A$ and $f_\xi$ (technically, the oracle
  $A\oplus f_\xi$), we get $\xi_A\le_T A\oplus f_\xi$ and therefore
  \[
    \xi_A[0,\infty)\oplus f_\xi \le_T \xi_A\oplus f_\xi
    \le_T (A\oplus f_\xi)\oplus f_\xi \equiv_T A\oplus f_\xi\,.
  \]

  We next show $A\le_T \xi_A[0,\infty)\oplus f_\xi$: To determine
  whether $s\in A$ suppose we already know which of the natural
  numbers $i<s$ belong to~$A$. Then the construction of $\xi_A$ above
  allows to build $t_s$ using the oracle $f_\xi$. Now construct
  $t_{s+1}$ assuming $s\in A$ again using the oracle $f_\xi$. If the
  resulting word
  \[
     u_0x_0w_0\,u_1x_1w_1\dots u_{s+1}x_{s+1}w_{s+1}
  \]
  (i.e., the ``second half'' of $z_{s+1}$) is a prefix of
  $\xi_A[0,\infty)$, then $s\in A$. Otherwise, $s\notin A$. Hence, indeed,
  $A\le_T \xi_A[0,\infty)\oplus f_\xi$ and therefore
  \[
    A\oplus f_\xi \le_T (f_A[0,\infty)\oplus f_\xi)\oplus f_\xi\equiv_T f_A[0,\infty)\oplus f_\xi\,.
  \]

  In summary, we showed
  $A\oplus f_\xi\equiv_T f_A[0,\infty)\oplus f_\xi$, the equivalence
  $A\oplus f_\xi\equiv_T f_A(-\infty,-1]\oplus f_\xi$ follows
  similarly.
\end{proof}

The following is the outcome of our attempt to relativise
Theorem~\ref{thm:non-recursive-but-decidable-MSOTh}. It differs in
several aspects from that theorem in that it talks about
Turing-degrees as opposed to m-degrees and that not every (large
enough) degree contains some MSO-equivalent bi-infinite word, but is
the join of the degree of such a word and an enumeration of the factor
set of $\xi$.

\begin{thm}\label{thm:Turing-degrees}
  Let $\xi$ be a recurrent non-periodic bi-infinite word, let
  $f_\xi\colon\N\to F(\xi)$ be a surjection, and let $\mathbf{a}$ be a
  Turing-degree above the degree of $f_\xi$.  Then there exists a
  bi-infinite word $\xi_A$ with $\MSOTh(\xi_A)=\MSOTh(\xi)$ such that
  $\xi_A(-\infty,-1]\oplus f_\xi \equiv_T \xi_A[0,\infty)\oplus
  f_\xi\in\mathbf{a}$.
\end{thm}

\begin{proof}
  Let $A\in\mathbf{a}$ be arbitrary and consider the bi-infinite word
  $\xi_A$ from Lemma~\ref{lem:type-A}. From $F(\xi)=F(\xi_A)$ and
  Theorem~\ref{thm:Semenov}, we get
  $\MSOTh(\xi)=\MSOTh(\xi_A)$. Furthermore, $f_\xi\le_T A$ implies
  $A\equiv_T A\oplus f_\xi$ and therefore
  $\xi_A(-\infty,-1]\oplus f_\xi \equiv_T \xi_A[0,\infty)\oplus
  f_\xi\equiv_T A\in\mathbf{a}$.
\end{proof}

As a consequence, we obtain that
Theorem~\ref{thm:non-recursive-but-decidable-MSOTh} holds for any
recurrent and non-periodic bi-infinite word with a decidable theory
(since the decidability of $\MSOTh(\xi)$ implies that $F(\xi)$ is
decidable and therefore recursively enumerable):

\begin{cor}\label{cor:Turing-degrees}
  Let $\xi$ be a recurrent and non-periodic bi-infinite word such that
  $F(\xi)$ is recursively enumerable. Then every non-trivial
  Turing-degree $\mathbf{a}$ contains some bi-infinite word
  $\xi_{\mathbf{a}}$ with $\MSOTh(\xi)=\MSOTh(\xi_{\mathbf{a}})$ and
  $\xi_{\mathbf{a}}(-\infty,-1]\equiv_T\xi_{\mathbf{a}}[0,\infty)\in\mathbf{a}$.
\end{cor}

\begin{proof}
  By the assumption on the set $F(\xi)$, there is a recursive
  surjection $f_\xi\colon\N\to F(\xi)$. Now the bi-infinite word
  $\xi_A$ from Theorem~\ref{thm:Turing-degrees} satisfies
  $\MSOTh(\xi_A)=\MSOTh(\xi)$ and
  $\xi_A[0,\infty)\equiv_T\xi_A[0,\infty)\oplus f_\xi\in\mathbf{a}$ as
  well as
  $\xi_A(-\infty,-1]\equiv_T\xi_A(-\infty,-1]\oplus
  f_\xi\in\mathbf{a}$ since $f_\xi$ is decidable.
\end{proof}

\begin{rem}
  B{\`{e}}s and C{\'{e}}gielski construct labeled linear orders
  $\lambda$ of order type $\omega\cdot(\Z,\le)$ with $\MSOTh(\lambda)$
  decidable such that the first-order theory of $(\lambda,x)$ is
  undecidable for any $x\in\lambda$ (such structures are called
  ``weakly maximal decidable'').

  Let $\xi_{\mathbf{a}}$ be the bi-infinite word from the above
  corollary.  Since $\xi_{\mathbf{a}}[i,\infty)$ and
  $\xi_{\mathbf{a}}[0,\infty)$ differ in a finite word, only, we get
  that also $\xi_{\mathbf{a}}[i,\infty)$ belongs to $\mathbf{a}$ for
  any $i\in\Z$. Consequently, the Turing-degree of the first-order
  theory of $\xi_{\mathbf{a}}[i,\infty)$, and therefore of
  $(\xi_{\mathbf{a}},i)$ is, for any $i\in\Z$, at least
  $\mathbf{a}$. Consequently, any decidable $\MSO$-theory $\MSOTh(\xi)$ of a
  recurrent bi-infinite word is realised by some weakly maximal
  decidable structure that is a labeled linear order of order type
  $(\Z,\le)$.
\end{rem}

\subsection{A characterization \`a la Semenov I}

\begin{defi}
  Let $\xi$ be a bi-infinite word. A pair of functions
  $(\indrec_\leftarrow,\indrec_\rightarrow)$ with
  $\indrec_\leftarrow,\indrec_\rightarrow\colon\Sent\to\Z\cup\{\top\}$
  is an \emph{indicator of recurrence for $\xi$} if the following hold
  for any sentence $\varphi\in\Sent$:
  \begin{itemize}
  \item if $\indrec_\leftarrow(\varphi)=\top$, then $\forall
    k\in\Z\,\exists i\le j\le k\colon \xi[i,j]\models\varphi$
  \item if $\indrec_\leftarrow(\varphi)\neq\top$, then $\forall i\le j\le
    \indrec_\leftarrow(\varphi)\colon \xi[i,j]\models\lnot\varphi$
  \item if $\indrec_\rightarrow(\varphi)=\top$, then $\forall k\in\Z\,\exists
    j\ge i\ge k\colon \xi[i,j]\models\varphi$
  \item if $\indrec_\rightarrow(\varphi)\neq\top$, then $\forall j\ge
    i\ge \indrec_\rightarrow(\varphi)\colon
    \xi[i,j]\models\lnot\varphi$
  \end{itemize}
\end{defi}
A bi-infinite word $\xi$ ``consists'' of an $\omega^*$-word
$\xi_\leftarrow$ and an $\omega$-word $\xi_\rightarrow$. Then, roughly
speaking, an indicator of recurrence for the \emph{bi-infinite} word
$\xi$ consists of a pair of indicators of recurrence, one for
$\xi_\leftarrow$ and one for $\xi_\rightarrow$ (which is not quite
true since $\indrec_\leftarrow(\varphi)>\indrec_\rightarrow(\varphi)$
is possible).

Therefore, the following characterisation is very similar to
Theorem~\ref{thm:semenov}, the only difference is that the condition
``$\xi$ is recursive'' is replaced by the more general one ``$\xi$ is
recursive or recurrent''.

\begin{thm}\label{thm:semenov-z}
  Let $\xi$ be a bi-infinite word. Then $\MSOTh(\xi)$ is decidable if
  and only if $\xi$ has a recursive indicator of recurrence and the
  bi-infinite word $\xi$ is recursive or recurrent.
\end{thm}

This theorem is an immediate consequence of
Proposition~\ref{prop:semenov-z1} and \ref{prop:semenov-z2} below.  We
first consider the case that $\xi$ is non-recurrent where the full
analogy to Theorem~\ref{thm:semenov} holds:

\begin{prop}\label{prop:semenov-z1}
  Let $\xi$ be a non-recurrent bi-infinite word. Then $\MSOTh(\xi)$ is
  decidable if and only if $\xi$ has a recursive indicator of recurrence
  and the bi-infinite word $\xi$ is recursive.
\end{prop}

\begin{proof}
  Let $\xi_\leftarrow=\xi(-\infty,-1]$ and
  $\xi_\rightarrow=\xi[0,\infty)$.

  First suppose that the $\MSO$-theory of $\xi$ is decidable. Then, by
  Theorem~\ref{thm:Semenov}, the $\MSO$-theories of $\xi_\leftarrow$
  and $\xi_\rightarrow$ are decidable. From Theorem~\ref{thm:semenov},
  we learn that the two $\omega$-words $\xi_\leftarrow^R$ and
  $\xi_\rightarrow$ are recursive and have recursive indicators of
  recurrence. Consequently, the bi-infinite word $\xi$ is recursive
  and has a recursive indicator of recurrence.

  Conversely suppose that $\xi$ is recursive and
  $(\indrec_\leftarrow,\indrec_\rightarrow)$ is a recursive indicator
  of recurrence. Clearly, the infinite words $\xi_\leftarrow$ and
  $\xi_\rightarrow$ are recursive. Furthermore, the function
  \[
     \indrec\colon\Sent\to\N\colon \varphi\mapsto
     \begin{cases}
       \top & \text{ if }\indrec_\rightarrow(\varphi)=\top\\
       \max(0,\indrec_\rightarrow(\varphi)) & \text{ otherwise}
     \end{cases}
  \]
  is a recursive indicator of recurrence for the $\omega$-word
  $\xi_\rightarrow$. Hence, by Theorem~\ref{thm:semenov},
  $\MSOTh(\xi_{\rightarrow})$ is decidable; the decidability of
  $\MSOTh(\xi_\leftarrow)$ can be shown analogously. Since
  $\xi=\xi_\leftarrow \xi_\rightarrow$, the $\MSO$-theory of $\xi$
  is decidable.
\end{proof}

\begin{prop}\label{prop:semenov-z2}
  Let $\xi$ be a recurrent bi-infinite word. Then $\MSOTh(\xi)$ is
  decidable if and only if $\xi$ has a recursive indicator of
  recurrence.
\end{prop}

\begin{proof}
  First suppose $\MSOTh(\xi)$ is decidable. Consider the function
  \[
     f\colon\Sent\to\N\cup\{\top\}\colon\varphi\mapsto
     \begin{cases}
       \top & \text{if }\exists w\in F(\xi)\colon w\models\varphi\\
       0 & \text{otherwise.}
     \end{cases}
  \]
  Since $\xi$ is recurrent, $(f,f)$ is an indicator of recurrence for
  $\xi$. From the decidability of $\MSOTh(\xi)$, we get that
  $\MSOTh(F(\xi))$ is decidable. But this implies that $f$ is
  computable. Hence $(f,f)$ is a recursive indicator of recurrence.

  Conversely, suppose $(\indrec_\leftarrow,\indrec_\rightarrow)$ is a
  recursive indicator of recurrence for $\xi$. Then, for
  $\varphi\in\Sent$, we can decide whether there exists $w\in F(\xi)$
  with $w\models\varphi$ (since $\xi$ is recurrent, this is the case
  if and only if $\indrec_\leftarrow(\varphi)=\top$). Thus,
  $\MSOTh(F(\xi))$ is decidable implying, by
  Theorem~\ref{thm:Semenov}, that $\MSOTh(\xi)$ is decidable.
\end{proof}

Since Theorem~\ref{thm:semenov-z} follows from the two above
propositions, its proof is completed.\qed

Let $\mathbf{a}$ be some m-degree above $\Rec$ and let
$\xi_{\mathbf{a}}$ be the word from
Theorem~\ref{thm:non-recursive-but-decidable-MSOTh} with decidable
MSO-theory. Then $\xi_{\mathbf{a}}[0,\infty)\in\mathbf{a}$ is not
recursive. Consequently, the MSO-theory of
$\xi_{\mathbf{a}}[0,\infty)$ is above $\mathbf{a}$ and therefore
undecidable. In other words, we have a bi-infinite word with decidable
MSO-theory such that the MSO-theory of its positive part is
undecidable. As a consequence to Theorem~\ref{thm:semenov-z}, we now
show that this cannot happen for recursive bi-infinite words.

\begin{cor}\label{cor:RabinThomas-z-recursive}
  Let $\xi$ be a recursive bi-infinite word with a decidable
  $\MSO$-theory. Then the $\MSO$-theories of
  $\xi_\leftarrow=\xi(-\infty,-1]$ and of
  $\xi_\rightarrow=\xi[0,\infty)$ are both decidable.
\end{cor}

\begin{proof}
  By Theorem~\ref{thm:semenov-z}, $\xi$ has a recursive indicator of
  recurrence $(\indrec_\leftarrow,\indrec_\rightarrow)$. Define the
  functions $f,g\colon\Sent\to\N\cup\{\top\}$ as follows:
  \begin{align*}
    f(\varphi)&=
    \begin{cases}
      \top & \text{ if }\indrec_\leftarrow(\varphi)=\top\\
      0 & \text{ if }\indrec_\leftarrow(\varphi)\ge0\\
      |\indrec_\leftarrow(\varphi)|-1 & \text{ otherwise}
    \end{cases}\\
    g(\varphi)&=
    \begin{cases}
      \top & \text{ if }\indrec_\rightarrow(\varphi)=\top\\
      0 & \text{ if }\indrec_\rightarrow(\varphi)<0\\
      \indrec_\rightarrow(\varphi) & \text{ otherwise}
    \end{cases}
  \end{align*}
  We claim that $f$ and $g$ are indicators of recurrence for the two
  $\omega$-words $\xi_\leftarrow^R$ and $\xi_\rightarrow$ (for
  notational simplicity, we only prove it for the $\omega$-word
  $\xi_\rightarrow$): Let $\varphi\in\Sent$.
  \begin{itemize}
  \item If $g(\varphi)=\top$, then
    $\indrec_\rightarrow(\varphi)=\top$. Hence, for all $k\ge0$ there
    exist $j\ge i\ge k$ with $\xi[i,j]\models\varphi$. But $\xi$ and
    $\xi_\rightarrow$ agree in the interval $[i,j]$.
  \item Suppose $g(\varphi)\neq\top$, i.e.,
    $\indrec_\rightarrow(\varphi)\in\Z$. Hence, for all natural numbers
    $j\ge i\ge \indrec_\rightarrow(\varphi)$, we have
    $\xi[i,j]\models\lnot\varphi$. This implies (as above) that all
    $j\ge i\ge g(\varphi)$ satisfy
    $\xi_\rightarrow[i,j]\models\lnot\varphi$.
  \end{itemize}
  Note that $\xi_\leftarrow^R$ and $\xi_\rightarrow$ are recursive
  $\omega$-words (this is the only place where we use that $\xi$ is
  recursive). Hence, by Theorem~\ref{thm:semenov}, the $\MSO$-theories
  of $\xi_\leftarrow^R$ and of $\xi_\rightarrow$ are both
  decidable.
\end{proof}

\subsection{A characterization \`a la Semenov II}
We return to the question when the $\MSO$-theory of a recurrent
bi-infinite word is decidable.

\begin{defi}
  Let $\xi$ be a bi-infinite word. A pair of functions
  $(\indrec_\leftarrow',\indrec_\rightarrow')$ with
  $\indrec_\leftarrow,\indrec_\rightarrow\colon\Sent\to\{0,1,\top\}$
  is a \emph{weak indicator of recurrence for $\xi$} if there exist
  $x,y\in\Z$ such that $\indrec_\leftarrow'$ is the weak indicator of
  recurrence for $\xi(-\infty,x]^R$ and $\indrec_\rightarrow'$ is the
  weak indicator of recurrence for $\xi[y,\infty)$.
\end{defi}
Differently from $\omega$-words, a bi-infinite word can have more than
one weak indicator of recurrence based on different reference points
$x$ and $y$.

\begin{thm}
  Let $\xi$ be a bi-infinite word. Then $\MSOTh(\xi)$ is decidable if
  and only if $\xi$ has a recursive weak indicator of recurrence and
  $\xi$ is recursive or recurrent.
\end{thm}

\begin{proof}
  Again, we have to handle the cases of recurrent and of non-recurrent
  words separately.

  So first let $\xi$ be non-recurrent. Suppose that $\MSOTh(\xi)$ is
  decidable. Then, by Theorem~\ref{thm:Semenov}, the infinite words
  $\xi_\leftarrow=\xi(-\infty,-1]$ and $\xi_\rightarrow=\xi[0,\infty)$
  have decidable MSO-theories. From Corollary~\ref{cor:semenov+}, we
  learn that $\xi_\leftarrow^R$ and $\xi_\rightarrow$ (and therefore
  $\xi$) are recursive with recursive weak indicators of recurrence
  $\indrec_\leftarrow'$ and $\indrec_\rightarrow'$. Hence the pair
  $(\indrec_\leftarrow',\indrec_\rightarrow')$ is a recursive weak
  indicator of recurrence for $\xi$ and the bi-infinite word $\xi$ is
  recursive.

  Suppose, conversely, that $\xi$ is recursive and
  $(\indrec_\leftarrow',\indrec_\rightarrow')$ is a recursive weak
  indicator of recurrence for $\xi$. Then there are $x,y\in\Z$ such
  that $\xi(-\infty,x]^R$ and $\xi[y,\infty)$ are recursive with
  recursive weak indicators of recurrence. Hence, by
  Corollary~\ref{cor:semenov+}, these two infinite words have
  decidable MSO-theories. Since $\xi[y,\infty)$ and
  $\xi_\rightarrow=\xi[0,\infty)$ only differ in a finite word, also
  $\MSOTh(\xi_\rightarrow)$ is decidable (and similarly for
  $\xi_\leftarrow=\xi(-\infty,-1]$). From
  $\xi=\xi_\leftarrow\,\xi_\rightarrow$, it follows that also
  $\MSOTh(\xi)$ is decidable.\bigskip

  We now consider the case that $\xi$ is recurrent. Then, by
  Theorem~\ref{thm:semenov-z}, there exists a recursive indicator of
  recurrence $(\indrec_\leftarrow,\indrec_\rightarrow)$ for
  $\xi$. Define the recursive functions
  $\indrec_\leftarrow',\indrec_\rightarrow'\colon\Sent\to\{0,1,\top\}$
  as follows:
  \[
    \indrec_\leftarrow'(\varphi)=
    \begin{cases}
      \top & \text{ if }\indrec_\leftarrow(\varphi)=\top\\
      0 & \text{ otherwise}
    \end{cases}
    \qquad
    \indrec_\rightarrow'(\varphi)=
    \begin{cases}
      \top & \text{ if }\indrec_\rightarrow(\varphi)=\top\\
      0 & \text{ otherwise}
    \end{cases}
  \]
  Since $\xi$ is recurrent, $\indrec_\leftarrow'$ is the weak
  indicator of recurrence for $\xi(-\infty,-1]^R$ and
  $\indrec_\rightarrow'$ is the one for $\xi[0,\infty)$.

  Conversely, suppose $(\indrec_\leftarrow',\indrec_\rightarrow')$ is
  a recursive weak indicator of recurrence for $\xi$. Since $\xi$ is
  recurrent, it is also an indicator of recurrence for $\xi$. Hence,
  by Theorem~\ref{thm:semenov-z}, $\MSOTh(\xi)$ is decidable.
\end{proof}

\subsection{A characterization \`a la Rabinovich-Thomas I}

\begin{defi}
  Let $\xi$ be a bi-infinite word, $u,v,w\in\{0,1\}^+$, $k\in\N$, and
  $H_\leftarrow,H_\rightarrow\subseteq\Z$ be infinite.
  \begin{itemize}
  \item The pair $(H_\leftarrow,H_\rightarrow)$ is a
    \emph{$k$-homogeneous factorisation of $\xi$ into $(u,v,w)$} if
    \begin{itemize}
    \item $\xi[i,j-1]\equiv_k u$ for all $i,j\in H_\leftarrow$ with
      $i<j$,
    \item $\xi[i,j-1]\equiv_k v$ for all $i\in H_\leftarrow$ and $j\in
      H_\rightarrow$ with $i<j$ and
    \item $\xi[i,j-1]\equiv_k w$ for all $i,j\in H_\rightarrow$ with $i<j$.
    \end{itemize}
  \item The pair $(H_\leftarrow,H_\rightarrow)$ is
    \emph{$k$-homogeneous for $\xi$} if it is a $k$-homogeneous
    factorisation of $\xi$ into some finite words $(u,v,w)$.
  \item Let $H_\leftarrow=\{h_i^-\mid i\in\N\}$ and
    $H_\rightarrow=\{h_i^+\mid i\in\N\}$ with $h_0^->h_1^->\dots$ and
    $h_0^+<h_1^+<\dots$. The pair $(H_\leftarrow,H_\rightarrow)$ is
    \emph{uniformly homogeneous for $\xi$} if, for all $k\in\N$, the
    pair $(\{h_i^-\mid i\ge k\},\{h_i^+\mid i\ge k\})$ is
    $k$-homogeneous for $\xi$.
  \end{itemize}
\end{defi}

Let $\xi$ be a bi-infinite word split into an $\omega^*$-word
$\xi_\leftarrow$ and an $\omega$-word $\xi_\rightarrow$.  As for any
$\omega$-word, there exists a uniformly homogeneous set
$H_\rightarrow$ for $\xi_\rightarrow$. Symmetrically, there exists a
set $H_\leftarrow\subseteq\widetilde{\N}$ that is ``uniformly
homogeneous'' for $\xi_\leftarrow$. Then the pair
$(H_\leftarrow,H_\rightarrow)$ is a uniformly homogeneous pair for
$\xi=\xi_\leftarrow \xi_\rightarrow$.

We will now see that Theorem~\ref{thm:RabinThomas} naturally extends
to \emph{recursive} bi-infinite words
(Theorem~\ref{thm:no-RabinThomas-non-recursive} below demonstrates
that it does not extend to non-recursive bi-infinite words).

\begin{thm}\label{thm:RabinThomas-z-recursive}
  A recursive bi-infinite word $\xi$ has a decidable $\MSO$-theory if
  and only if there exists a recursive uniformly homogeneous pair for
  $\xi$.
\end{thm}

\begin{proof}
  Suppose $\MSOTh(\xi)$ is decidable. Consider the infinite words
  $\xi_\leftarrow=\xi(-\infty,-1]$ and
  $\xi_\rightarrow=\xi[0,\infty)$. By
  Corollary~\ref{cor:RabinThomas-z-recursive}, the $\MSO$-theories of
  $\xi_\leftarrow^R=\xi(-\infty,-1]^R$ and of
  $\xi_\rightarrow=\xi[0,\infty)$ are both decidable. Consequently, by
  Theorem~\ref{thm:RabinThomas}, there are recursive uniformly
  homogeneous factorisations $H_\leftarrow^R,H_\rightarrow\subseteq\N$
  for $\xi_\leftarrow^R$ and $\xi_\rightarrow$ into $(x^R,y^R)$ and
  $(y',z)$, respectively.  Deleting, if necessary, the minimal element
  from $H_\leftarrow^R$, we can assume $0\notin H_\leftarrow^R$. We
  set $H_\leftarrow=\{-n\mid n\in
  H_\leftarrow^R\}\subseteq\widetilde{\N}$ and show that
  $(H_\leftarrow,H_\rightarrow)$ is a uniformly homogeneous pair for
  $\xi$: Let $H_\leftarrow=\{h_i^-\mid i\in\N\}$ and
  $H_\rightarrow=\{h_i^+\mid i\in\N\}$ such that $h_0^->h_1^->\dots$
  and $h_0^+<h_1^+<\dots$.
  \begin{itemize}
  \item Let $j>i\ge k$. Then
    \begin{align*}
      \xi[h_j^-,h_i^--1]
       &= \xi_\leftarrow[h_j^-,h_i^--1] \\
       &= (\xi_\leftarrow^R[|h_i^-|,|h_j^-|-1])^R \\
       &= (\xi_\leftarrow^R[-h_i^-,-h_j^--1])^R \\
       &\equiv_k y^R
    \end{align*}
    since $-h_i^-,-h_j^i\in H_\leftarrow^R$ and $k\le i<j$.
  \item Let $i,j\ge k$. Then
    \begin{align*}
      \xi[h_i^-,h_j^+-1]
        &= \xi_\leftarrow[h_i^-+1,0]\,\xi_\rightarrow[0,h_j^+-1]\\
        &= (\xi_\leftarrow^R[0,|h_i^-+1|])^R\,\xi_\rightarrow[0,h_j^+-1]\\
        &= (\xi_\leftarrow^R[0,-h_i^--1])^R\,\xi_\rightarrow[0,h_j^+-1]\\
        &\equiv_k x^Ry'
    \end{align*}
    since $-h_i^-\in H_\leftarrow^R$, $h_j^+\in H_\rightarrow$, and
    $i,j\ge k$.
  \item Let $j>i\ge k$. Then
    $\xi[h_i^+,h_j^+-1]=\xi_\rightarrow[h_i^+,h_j^+-1]\equiv_k z$.
  \end{itemize}
  Hence the pair $(\{h_i^-\mid i\ge k\},\{h_i^+\mid i\ge k\})$ is a
  $k$-homogeneous factorisation of $\xi$ into $(y^R,x^Ry',z)$. Since
  $k$ is arbitrary, $(H_\leftarrow,H_\rightarrow)$ is uniformly
  homogeneous for $\xi$. Since these two sets are clearly recursive,
  this proves the first implication.

  Conversely, suppose there exists a recursive uniformly homogeneous
  pair $(H_\leftarrow,H_\rightarrow)$ for $\xi$. Then the sets
  $H_\leftarrow^R=\{|n|\mid n\in H_\leftarrow\cap\widetilde{\N}\}$ and
  $H_\rightarrow\cap\N$ are recursive and uniformly homogeneous for
  $\xi_\leftarrow^R$ and $\xi_\rightarrow$, resp. Since
  $\xi_\leftarrow$ and $\xi_\rightarrow$ are both recursive, we can
  apply Theorem~\ref{thm:RabinThomas}. Hence the infinite words
  $\xi_\leftarrow$ and $\xi_\rightarrow$ both have decidable
  $\MSO$-theories. Since $\xi=\xi_\leftarrow \xi_\rightarrow$, the
  $\MSO$-theory of $\xi$ is decidable.
\end{proof}

We next show that we cannot hope to extend the characterisation from
Theorem~\ref{thm:RabinThomas-z-recursive} to non-recursive words. The
counterexample we construct is simplest possible (namely, recursively
enumerable) and does not even have a uniformly homogeneous pair that
is recursively enumerable.

\begin{thm}\label{thm:no-RabinThomas-non-recursive}
  There exists a recurrent recursively enumerable bi-infinite word
  $\xi$ with decidable $\MSO$-theory such that there is no recursively
  enumerable uniformly homogeneous pair for $\xi$.
\end{thm}

\begin{proof}
  We prove this theorem by constructing a recurrent bi-infinite word
  $\xi$ such that the set $F(\xi)$ of factors is $\{0,1\}^*$. Hence
  $\xi$ has decidable $\MSO$-theory by Theorem~\ref{thm:Semenov}.

  Let $e,s\in\N$ and define the function $g_{e,s}\colon\N\to\N$ by
  \[
     g_{e,s}(n)=
     \begin{cases}
       1 & \text{if $n\le s$ and the computation of $\Phi_e(n)$ halts in $\le s$ steps}\\
       0 & \text{otherwise.}
     \end{cases}
  \]
  The function $g_{e,s}$ is computable and, even more, from $e$ and
  $s$, one can compute an index $f(e,s)$ such that
  $g_{e,s}=\Phi_{f(e,s)}$. With
  $W_{e,s}=\{n\in\N\mid \Phi_{f(e,s)}(n)=1\}$, we get
  \begin{itemize}
  \item $\Phi_{f(e,s)}$ is total,
  \item $W_{e,s}\subseteq\{0,1,\dots,s\}$, and
  \item $W_e=\bigcup_{s\in\N}W_{e,s}$
  \end{itemize}

  Furthermore, we fix some recursive enumeration $u_0,u_1,\dots$ of
  the set $\{0,1\}^+$ of non-empty finite words.

  \paragraph{\underline{Construction}}
  By induction on $s\in\N$, we construct tuples
  \[
   t_s= (w_s,m_{0,s},m_{1,s},\dots,m_{s,s},P_s)\in
       \{0,1\}^*\times\N^{s+1}\times2^{\{0,\dots,s\}}
  \]
  such that
  \begin{itemize}
  \item $m_{i,s}+|u_i|\le m_{i+1,s}$ for all $0\le i<s$ and
    $m_{s,s}+|u_s|\le|w_s|$ (in particular, $|w_s|>s$),
  \item $w_s[m_{i,s},m_{i,s}+|u_i|-1]=u_i$ for all $0\le i\le s$, and
  \item for all $e\in P_s$, there exist $a,b\in W_e$ with $a< b<|w_s|$
    and $w_s[a,b-1]\in 1^*$.
  \end{itemize}
  In other words, the finite word $w_s$ contains disjoint occurrences
  of the factors $u_0,u_1,\dots,u_s$ at positions
  $m_{0,s},m_{1,s},\dots,m_{s,s}$ and a factor from $1^*$ between two
  positions from $W_e$ (for $e\in P_s$).

  At the beginning, set $w_0=u_0$, $m_{0,0}=0$, and $P_0=\emptyset$. Then
  the inductive invariant holds for the tuple $t_0=(w_0,m_{0,0},P_0)$.

  Now suppose the tuple $t_s$ has been constructed. Let $H_{s+1}$
  denote the set of indices $0\le e\le s+1$ with $e\notin P_s$ such
  that $W_{e,s}$ contains at least two numbers $b>a\ge m_{e,s}$. In
  the construction of the tuple $t_{s+1}$, we distinguish two cases:
  \begin{itemize}
  \item 1st case: $H_{s+1}=\emptyset$. Then set $w_{s+1}=w_su_{s+1}$,
    $m_{i,s+1}=m_{i,s}$ for $0\le i\le s$, $m_{s+1,s+1}=|w_s|$, and
    $P_{s+1}=P_s$. Since the inductive invariant holds for the tuple
    $t_s$, it also holds for the newly constructed tuple $t_{s+1}$.
  \item 2nd case: $H_{s+1}\neq\emptyset$. Let $e_{s+1}$ be the minimal
    element of $H_{s+1}$ and let $a_{s+1}$ and $b_{s+1}$ be the
    minimal elements of $W_{e_{s+1},s}$ satisfying $m_{e,s}\le
    a_{s+1}<b_{s+1}$. Then set
    \begin{itemize}
    \item $w_{s+1}=w_s[0,a_{s+1}-1]\, 1^{b_{s+1}-a_{s+1}}\,
      w_s[b_{s+1},|w_s|-1]\, u_{e_{s+1}} u_{e_{s+1}+1} \dots u_{s+1}$
      (in other words, the words $u_{e_{s+1}}$ up to $u_{s+1}$ are
      appended to $w_s$ and the positions between $a_{s+1}$ and
      $b_{s+1}-1$ are set to $1$).
    \item $m_{i,s+1}=
      \begin{cases}
        m_{i,s} & \text{ if } i<e_{s+1}\\
        |w_s u_{e_{s+1}} u_{e_{s+1}+1}\dots u_{i-1}| & \text{ if }e_{s+1}\le i\le s+1
      \end{cases}$
    \item $P_{s+1}=P_s\cup\{e_{s+1}\}$
    \end{itemize}

    The first two conditions of the inductive invariant are
    obvious. Regarding the last one, let $e\in P_{s+1}$. If $e\neq
    e_{s+1}$, then $e\in P_s$ and therefore there exist $a,b\in W_e$
    with $a< b<|w_s|<|w_{s+1}|$ such that $w_s[a,b-1]\in 1^*$. Note
    that any position in $w_s$ that carries $1$ also carries $1$ in
    $w_{s+1}$. Hence $w_{s+1}[a,b-1]\in 1^*$ as well. It remains to
    consider the case $e=e_{s+1}$. But then, by the very construction,
    $a_{s+1}<b_{s+1}$ belong to $W_{e_{s+1},s}\subseteq W_e$ and
    satisfy $w_{s+1}[a_{s+1},b_{s+1}-1]\in 1^*$.
  \end{itemize}
  This finishes the construction of the sequence of tuples $t_s$.

  Let $\xi_\rightarrow$ be the $\omega$-word with
  $\xi_\rightarrow(i)=1$ iff there exists $s\in\N$ with $w_s(i)=1$.

  \paragraph{\underline{Claim 1}} The $\omega$-word $\xi_\rightarrow$ is
  recursively enumerable.

  \paragraph{Proof of Claim 1} Note that the tuple $t_{s+1}$ is
  computable from the tuple $t_s$. \hspace{\fill}\textbf{q.e.d.}
  \medskip

  \paragraph{\underline{Claim 2}} The $\omega$-word $\xi_\rightarrow$ is rich,
  i.e., any finite word is a factor of $\xi_\rightarrow$.

  \paragraph{Proof of Claim 2}
  Let $u\in\{0,1\}^+$. Then there exists $e\in\N$ with $u=u_e$.  Note
  that $m_{e,s}\le m_{e,s+1}$ for all $e,s\in\N$. Furthermore,
  $m_{e,s} < m_{e,s+1}$ iff $H_{s+1}\neq\emptyset$ and $e_{s+1}\le
  e$. Since the numbers $e_{s'+1}$ for $s'\in\N$ (if defined) are
  mutually distinct, there exists $s\in\N$ such that $e_{t+1}>e$ and
  therefore $m_{e,s}=m_{e,t}$ for all $t\ge s$.

  Consequently,
  $\xi_\rightarrow[m_{e,s},m_{e,s}+|u_e|-1]=w_s[m_{e,s},m_{e,s}+|u_e|-1]=u_e
  =u$.  \hspace{\fill}\textbf{q.e.d.}  \medskip

  \paragraph{\underline{Claim 3}} If $W_e$ is infinite, then
  $e\in\bigcup_{s\in\N}P_s$.

  \paragraph{Proof of Claim 3} By contradiction, suppose this is not
  the case. Let $e\in\N$ be minimal with $W_e$ infinite and
  $e\notin\bigcup_{s\in\N}P_s$. Since $W_e$ is infinite, we get $e\in
  H_{s+1}$ for almost all $s\in\N$. Since $e$ was chosen minimal,
  there exists $s\in\N$ with $e=\min H_{s+1}$. But then $e_{s+1}=e$
  and therefore $e\in P_{s+1}$.\hspace*{\fill}\textbf{q.e.d.}\medskip

  \paragraph{\underline{Claim 4}} No recursively enumerable set $W$ is uniformly
  homogeneous for the $\omega$-word $\xi_\rightarrow$.

  \paragraph{Proof of Claim 4} Suppose $W$ is recursively enumerable
  and uniformly homogeneous for $\xi_\rightarrow$. Then $W$ is
  infinite and there exists $e\in\N$ with $W=W_e$. By claim 3, there
  exists $s\in \N$ with $e\in P_s$. Hence there are $a,b\in W_e$ with
  $w_s[a,b-1]\in 1^*$ and therefore
  $\xi_\rightarrow[a,b-1]=w_s[a,b-1]$. By claim 2, there are $d>c>b$
  in $W_e$ such that $\xi_\rightarrow[c,d-1]\notin 1^*$. But then
  $\xi_\rightarrow[a,b-1]$ and $\xi_\rightarrow[c,d-1]$ do not have
  the same $1$-type. Hence the set $W_e$ is not 1- and therefore not
  uniformly homogeneous for
  $\xi_\rightarrow$. \hspace{\fill}\textbf{q.e.d.}\medskip

  Finally, let $\xi_\leftarrow$ be the reversal of $\xi_\rightarrow$
  and consider the bi-infinite word
  $\xi=\xi_\leftarrow\,\xi_\rightarrow$. By
  Theorem~\ref{thm:semenov-z}, $\MSOTh(\xi)$ is decidable since $\xi$
  is recurrent and contains every finite word as a factor. It is
  recursively enumerable by claim 1. Finally, suppose
  $(H_\leftarrow,H_\rightarrow)$ is uniformly homogeneous
  for~$\xi$. Then $H_\rightarrow\cap\N$ is uniformly homogeneous for
  $\xi_\rightarrow$. By claim 4, this set cannot be recursively
  enumerable. Hence $(H_\leftarrow,H_\rightarrow)$ is not recursively
  enumerable either.
\end{proof}

\subsection{A characterization \`a la Rabinovich-Thomas II}

We next extend the 2nd characterisation by Rabinovich and Thomas
(Theorem~\ref{thm:RabinThomas2}) to bi-infinite words. Differently
from the first characterization, this will also cover non-recursive
bi-infinite words.

\begin{defi}
  Let $\xi$ be some bi-infinite word and
  $\typefun\colon\N\to\{0,1\}^+\times\{0,1\}^+\times\{0,1\}^+$. The
  function $\typefun$ is a \emph{type-function for $\xi$} if, for all
  $k\in\N$, the bi-infinite word $\xi$ has a $k$-homogeneous
  factorisation into $\typefun(k)$.
\end{defi}

We will show that the $\MSO$-theory of a bi-infinite word is decidable
if and only if it has a recursive type-function.

\begin{thm}\label{thm:RabinThomas2-z}
  Let $\xi$ be a bi-infinite word. Then $\MSOTh(\xi)$ is decidable if
  and only if $\xi$ has a recursive type-function.
\end{thm}

\begin{proof}
  First suppose that $\MSOTh(\xi)$ is decidable.  We have to construct
  a recursive type-function $\typefun\colon\N\to(\{0,1\}^+)^3$. To
  this aim, let $k\in\N$. Then one can compute a finite sequence
  $\varphi_1,\dots,\varphi_n$ of $\MSO$-sentences of quantifier-rank
  $k$ such that, for all finite words $u$ and $v$, we have $u\equiv_k
  v$ if and only if
  \[
     \forall 1\le i\le n\colon u\models\varphi_i\iff v\models\varphi_i\,.
  \]
  For finite words $u$, $v$, and $w$, consider the following
  statement:
  \[
    \exists H_\leftarrow,H_\rightarrow\colon
    \begin{array}[t]{cl}
      & \forall y\colon(\exists x,z\colon x<y<z\land H_\leftarrow(x)\land H_\rightarrow(z))\\
      \land &
      \forall x,z\colon(H_\leftarrow(x)\land H_\rightarrow(z)\to x<z)\\
      \land &
      \forall x,y\colon(x<y\land H_\leftarrow(x)\land H_\leftarrow(y)\to \xi[x,y-1]\equiv_k u)\\
      \land &
      \exists x,y\colon
        (H_\leftarrow(x)\land H_\rightarrow(y)
        \land  \xi[x,y-1]\equiv_k v)\\
      \land &
      \forall x,y\colon(x<y\land H_\rightarrow(x)\land H_\rightarrow(y)\to \xi[x,y-1]\equiv_k w)\\
    \end{array}
  \]
  This statement holds for a bi-infinite word $\xi$ if and only if
  $\xi$ has a $k$-homogeneous factorisation into $(u,v,w)$.  Using the
  $\MSO$-sentences $\varphi_1,\dots,\varphi_n$, the statements
  $\xi[x,y-1]\equiv_k u$ etc.\ can be expressed as $\MSO$-formulas
  with free variables $x$ and $y$. Since the $\MSO$-theory of $\xi$ is
  decidable, we can therefore decide (given $k$, $u$, $v$, and $w$)
  whether $\xi$ has a $k$-homogeneous factorisation into
  $(u,v,w)$. Since some $k$-homogeneous factorisation always exist,
  this allows to compute, from $k$, a tuple $\typefun(k)$ such that
  $\xi$ has a $k$-homogeneous factorisation into $\typefun(k)$. Thus,
  we obtained a recursive type-function $\typefun$.

  Conversely suppose that $\typefun$ is a recursive type-function for
  $\xi$. To show that $\MSOTh(\xi)$ is decidable, let
  $\varphi\in\Sent$ be any $\MSO$-sentence. Let $k$ denote the
  quantifier-rank of $\varphi$. First, compute
  $\typefun(k)=(u,v,w)$. Then $\xi\models\varphi$ iff $ u^{\omega^*} v
  w^\omega\models\varphi$ which is decidable since this
  bi-infinite word is ultimately periodic on the left and on the
  right.
\end{proof}

\section{How many MSO-equivalent bi-infinite words are there?}
\label{sec:counting}

If $\alpha$ and $\beta$ are $\omega$-words and $\MSO$-equivalent, then
$\alpha=\beta$. In this final section we study this question for
bi-infinite words. Shift-equivalence and period will be important
notions in this context.

To count the number of $\MSO$-equivalent bi-infinite words, we need a
characterisation when two bi-infinite words are $\MSO$-equivalent.

\begin{thm}\label{thm:PerrinPin}\cite[Chp.\ 9, Theorem\ 6.1]{PerrinPin}
  Two bi-infinite words $\xi$ and $\zeta$ are $\MSO$-equivalent if and
  only if one of the following conditions is satisfied:
  \begin{enumerate}
  \item $\xi$ and $\zeta$ are shift-equivalent.
  \item $\xi$ and $\zeta$ are recurrent and have the same set of
    factors.
  \end{enumerate}
\end{thm}

This characterisation is the central ingredient in the proof of the
main result of this final section:
\begin{thm}\label{thm:type}
  Let $\xi$ be a bi-infinite word.
  \begin{enumerate}
  \item[(a)] If $\xi$ is periodic, then the cardinality of the type of
    $\xi$ is finite and equals the period of $w$.
  \item[(b)] If $\xi$ is non-recurrent, then the cardinality of the
    type of $\xi$ is $\aleph_0$.
  \item[(c)] If $\xi$ is recurrent and non-periodic, then the
    cardinality of the type of $\xi$ is $2^{\aleph_0}$.
  \end{enumerate}
\end{thm}

\begin{proof}
  \begin{enumerate}[label=(\alph*)]
  \item Let $p$ be the period of $\xi$. Since $p$ is minimal, there
    are precisely $p$ distinct bi-infinite words that are
    shift-equivalent with~$\xi$. Since shift-equivalent words are
    $\MSO$-equivalent, the type of $\xi$ contains at least $p$
    elements. It remains to be shown that no further $\MSO$-equivalent
    word exists. So let $\zeta$ be some $\MSO$-equivalent word. Then
    $\zeta$ is $p$-periodic since $\xi$ (and therefore $\zeta$)
    satisfies $\forall x\colon(P(x)\Leftrightarrow P(x+p))$ and does
    not satisfy $\forall x\colon(P(x)\Leftrightarrow P(x+q))$ for any
    $1\le q<p$. Furthermore $u=\xi[1,p]$ is a factor of $\xi$ and
    therefore of $\zeta$ of length $p$. Hence $\zeta=u^{\omega^*}
    u^\omega$.
  \item This claim follows immediately from Theorem~\ref{thm:PerrinPin}.
  \item Above any Turing-degree, there are $2^{\aleph_0}$
    Turing-degrees. Hence the claim follows from
    Theorem~\ref{thm:Turing-degrees}.
  \end{enumerate}
\end{proof}

\end{document}